\documentclass{article}

\usepackage[a4paper]{geometry} 
\usepackage[T1]{fontenc}							
\usepackage{lmodern}
\usepackage{amsthm}
\usepackage[utf8]{inputenc}							
\usepackage{amsfonts, amssymb, amsmath, amsbsy, amsbsy}
\usepackage[fixamsmath,disallowspaces]{mathtools}
\usepackage{epsfig}
\usepackage{graphicx}
\usepackage{colordvi}
\usepackage{pgfplots}
\pgfplotsset{compat=1.12}
\usetikzlibrary{circuits.ee.IEC}
\usepackage{circuitikz}
\usepackage{booktabs}         
\usepackage{array}            
\usepackage{paralist}         
\usepackage{verbatim}         

\usepackage{microtype}				
\usepackage{subfigure} 
\usepackage{hyperref}
\usepackage{cleveref}					
\usepackage{tikz}
\usepackage{siunitx}
\usepackage{siunitx}

\date{January 2024}

\newcommand{\be}{\begin{equation}}
\newcommand{\ee}{\end{equation}}
\newcommand{\bea}{\begin{eqnarray}}
\newcommand{\eea}{\end{eqnarray}}
\newcommand{\beaa}{\begin{eqnarray*}}
\newcommand{\eeaa}{\end{eqnarray*}}
\newtheorem{lemma}{Lemma}
\newtheorem{theorem}{Theorem}

\newcommand{\n}{{n}}
\newcommand{\p}{{p}}

\newcommand{\cc}{{c}}
\newcommand{\vv}{{v}}

\definecolor{darkred}{rgb}{0.6,0,0}

\title{Forward lateral photovoltage scanning problem: Perturbation approach and existence-uniqueness analysis}

\author{
    Giuseppe Alì \\
    Department of Physics, University of Calabria \\
    via Pietro Bucci 30/B, Arcavacata di Rende, I-87036, Cosenza, Italy \\
    {\tt giuseppe.ali@unical.it}
    \and
    Patricio Farrell \\
    Weierstrass Institute Berlin \\
    Mohrenstr. 39, 10117 Berlin, Germany \\
    {\tt patricio.farrell@wias-berlin.de}
    \and
    Nella Rotundo \\
    Dipartimento di Matematica e Informatica 'Ulisse Dini' \\
    University of Florence \\
    Viale Morgagni 67/A, 50134 Florence, Italy \\
    {\tt nella.rotundo@unifi.it}
}

\date{} 

\begin{document}

\maketitle

\begin{abstract}
In this paper, we present analytical results for the so-called forward lateral photovoltage scanning (LPS) problem. The (inverse) LPS  model predicts doping variations in crystal by measuring the current leaving the crystal generated by a laser at various positions. The forward model consists of a set of nonlinear elliptic equations coupled with a measuring device modeled by a resistance. Standard methods to ensure the existence and uniqueness of the forward model cannot be used in a straightforward manner due to the presence of an additional generation term modeling the effect of the laser on the crystal. Hence, we scale the original forward LPS problem and employ a perturbation approach to derive the leading order system and the correction up to the second order in an appropriate small parameter. While these simplifications pose no issues from a physical standpoint, they enable us to demonstrate the analytic existence and uniqueness of solutions for the simplified system using standard arguments from elliptic theory adapted to the coupling with the measuring device.
\end{abstract}



\emph{Keywords}:
Drift-diﬀusion model,
Charge transport, 
Lateral photovoltage scanning method (LPS),
Perturbation analysis, 
Existence and uniqueness.

\section{Introduction}




Estimating crystal inhomogeneities in semiconductors noninvasively holds significance for various industrial applications. For example, understanding and avoiding doping fluctuations is important if very pure semiconductor crystal are required in industrial applications \cite{Zulehner1994}. However, such fluctuations, often referred to as striations, may not always have a negative impact which may be the case for grain boundaries in solar cells \cite{Randle2010}.  

Doping variations lead to local electrical fields. In order to detect  inhomogeneities, variations, and defects in the doping profile noninvasively, one may systematically generate electron-hole pairs through some form of electromagnetic radiation and attach an external circuit to measure the resulting current. Due to the local fields created by doping inhomogeneities, charge carriers tend to redistribute in the surrounding region of the excitation to minimize energy. The excess charge carriers then traverse the external circuit, generating a measurable current.
Scanning the semiconductor sample with the electromagnetic source at different positions allows one to visualize the distribution of electrically active charge-separating defects and variations in the doping profile along the scan locations.
Several different types of technologies rely on this physical principle such as 
Electron Beam Induced Current (EBIC)~\cite{Wittry1967}, Laser Beam Induced
Current (LBIC)~\cite{Bajaj1987}, scanning photovoltage
(SPV)~\cite{Jastrzebski1982}, and Lateral Photovoltage Scanning
(LPS)~\cite{Luedge1997, Kayser2020b, Kayser2021,Manganelli2022}. 
While EBIC uses localized electron beams as the source of electromagnetic radiation, LBIC, SPV and LPS use laser beams. 

In this paper, we focus on the LPS method. A nonlinear drift-diffusion model for the forward problem was mathematically formulated in \cite{Farrell2021}, where the coupling between the charge transport in the crystal and the external circuit for the voltmeter is realized through an implicit boundary condition. This forward model computes current/voltages at the contacts for given doping variations. However, no analytical existence results were presented. The main difficulty in establishing the existence is due to the generation rate on the right-hand side of the continuity equations which impedes the direct use of standard arguments \cite{Ali2010,Markowich}. Known existence results in the literature for related PDE models require the generation term to be small \cite{Busenberg1993}. However, the authors in \cite{Farrell2021} demonstrated through simulations based on the Voronoi finite-volume method \cite{Farrell2017,Patriarca2018,Farrell2017c,Koprucki2015,Abdel2020b} that, on the one hand, numerically the system also converges for higher laser powers and that, on the other hand, the electron density is several orders of magnitude larger than the hole density. In this paper, we make use of the second fact by deriving an appropriately scaled  unipolar model and identifying the relevant dimensionless small parameters. However, setting the smallest parameter to zero, oversimplifies the physics in the sense that the coupling between the circuit and the charge transport model is ignored. Therefore, through a perturbation approach the leading order corrections up to the second order in the small parameter are derived. The resulting system is simpler than the original LPS model. However, for realistic applications, for example studied in \cite{Farrell2021,Luedge1997, Kayser2020b}, these simplifications are reasonable since the perturbation parameters for silicon and gallium arsenide are small. Thus, these simplifications are not problematic from a physical point of view. Moreover, from a mathematical point of view, we have the additional benefit that the leading order model is a decoupled nonlinear elliptic system and the second order correction is a linear elliptic system coupled with the external circuit. The existence and uniqueness of solutions for the simplified model can be shown by using arguments from elliptic PDE theory modified to include the coupling with the external circuit. While this paper focuses on the mathematical analysis of the forward problem, it is worth noting that numerical, data-driven strategies exist to address the more complicated inverse problem \cite{Piani2022}, that is, the prediction of the doping variations in a crystal by measuring the current generated by a laser beam impinging at various positions of the crystal.

The remainder of this paper is organized as follows: In Section \ref{sec:basic_model}, we state the original LPS model. The model is then scaled in Section \ref{sec:scaling} and the asymptotic simplified model is derived in Section \ref{sec:expansion}. In Section \ref{sec:existence}, we show the existence and uniqueness of solutions for the asymptotic PDE system and draw some conclusions in Section \ref{sec:conclusion}.

\section{Lateral Photovoltage Scanning method}
\label{sec:basic_model}
In this section, we present the model for the LPS method.
The description of the model strongly relies on \cite{Farrell2021}. 

\subsection{Presentation of the model}
We represent the silicon crystal as a confined domain $\Omega\subset \mathbb{R}^d$, where $d=1,2,3$. The doping profile, denoted by $C(\boldsymbol{x})$ for $\boldsymbol{x}\in \Omega$, reflects the difference between donor and acceptor concentrations.

Within the crystal, two charge carriers exist: electrons with a negative elementary charge $-q$ and holes with a positive elementary charge $q$. Charge transport is described by the electrostatic potential $\psi(\boldsymbol{x})$ and quasi Fermi potentials $\varphi_{\n}(\boldsymbol{x})$ and $\varphi_{\p}(\boldsymbol{x})$ for electrons and holes, respectively, which are governed by the steady-state van Roosbroeck model:
\begin{equation}
  \begin{split}
-\nabla \cdot (\varepsilon \nabla \psi) &= 
q(p-n+ C(\boldsymbol{x}))
\\
- \nabla\cdot ( \mu_\n  n \nabla \varphi_\n ) &= R -G(\boldsymbol{x}), \qquad
\\
- \nabla\cdot (\mu_\p  p \nabla \varphi_{\p}) &= G(\boldsymbol{x})-R. \qquad 
\end{split}
\label{eq:vR-model}
\end{equation}
The first equation describes a self-consistent electric field through a nonlinear Poisson equation. The subsequent continuity equations describe charge transport in the crystal, with constant permittivity $\varepsilon$ and mobilities $\mu_n$ and $\mu_p$. Assuming Boltzmann statistics, the relations between electron and hole densities $n$ and $p$, and the quasi-Fermi potentials, are established by:
\begin{equation}
\begin{split}
 n(\psi,\varphi_\n)&= N_\cc \exp \left( \frac{q(\psi - \varphi_\n)-E_{\cc}}{k_B T}\right),
\\
 p(\psi,\varphi_\p)&= N_\vv \exp \left( \frac{ q(\varphi_\p - \psi) + E_{\vv}}{k_B T}\right).
\label{eq:dens-pot}
\end{split}
\end{equation}

The effective conduction and valence band densities of states are denoted by $N_\cc$ and $N_\vv$, while the Boltzmann constant and temperature are represented by $k_B$ and $T$. The constant conduction and valence band-edge energies are denoted as $E_{\cc}$ and $E_{\vv}$.

The current densities for electrons and holes, $\boldsymbol{J}_\n(\boldsymbol{x})$ and $\boldsymbol{J}_\p(\boldsymbol{x})$, are expressed as $-q \mu_\n n \nabla \varphi_\n$ and $-q \mu_\p p \nabla \varphi_{\p}$, respectively. Utilizing the relations \eqref{eq:dens-pot}, their drift-diffusion form is
\begin{equation}
\boldsymbol{J}_\n = -q \mu_\n  ( n \nabla \psi- V_{\mathrm{th}} \nabla n), \qquad
 \boldsymbol{J}_\p =   -q \mu_\p  (p  \nabla \psi +  V_{\mathrm{th}} \nabla p ),
 \label{eq:curr-dd}
\end{equation}
where $V_{\mathrm{th}} = k_BT/q$ is the thermal voltage. 

Recombination $R$ and generation rates $G$ are elaborated in subsequent sections. 

The system \eqref{eq:vR-model} is complemented by mixed boundary conditions. The boundary $\partial\Omega$ comprises two disjoint parts, $\Gamma_N$ and $\Gamma_D$. Neumann boundary conditions on $\Gamma_N$ are given by
\begin{equation}
\frac{\partial \psi}{\partial \boldsymbol{\nu}}=
\frac{\partial \varphi_\n}{\partial \boldsymbol{\nu}}=
\frac{\partial \varphi_\p}{\partial \boldsymbol{\nu}}=0,  \quad \text{on }\Gamma_N,
\label{eq:neumann-data}
\end{equation}
where $\partial/\partial \boldsymbol{\nu} =\boldsymbol{\nu} \cdot \nabla$ is the normal derivative along the external unit normal $\boldsymbol{\nu}$.
On $\Gamma_D$, Dirichlet-type boundary conditions model ohmic contacts
\begin{equation}
\psi - \psi_0= \varphi_\n = \varphi_\p = u_{D_i} - u_{\text{ref}}\quad \text{on }\quad\Gamma_{D_i}, i=1,2,
\label{eq:Dirichlet-data}
\end{equation}
where $ \psi_0$ is the local electroneutral potential (see \cite{Markowich} for details),
\be
\psi_0 := u_{D_i}+\frac{E_{\cc}}{q}+\frac{k_B T}{q}\log\left(\frac{n_0}{N_\cc}\right),
\quad
n_0:=\frac{1}{2}\left(C+\sqrt{C^2+4n_{\text{i}}^2}\right).
\ee
The parameter $n_{\text{i}}$ refers to the intrinsic carrier density defined by
\begin{equation}
  n_{\text{i}}^2 = N_c N_v\exp\left(-\frac{E_c-E_v}{k_B T}\right).
\label{eq:intrinsic_densitiy}
\end{equation} 
The terms $u_{D_i}$ represent contact voltages at the respective ohmic contacts, with a common practice of defining a reference potential $u_{\text{ref}}$.
The total electric current $j_{D_i}$ flowing through the $i$-th ohmic contact $\Gamma_{D_i}$ is defined by the surface integral
\begin{equation}
j_{D_i}  =-\int_{\Gamma_{D_i}}   \boldsymbol{\nu} \cdot (\boldsymbol{J}_\n (\boldsymbol{x})+\boldsymbol{J}_\p (\boldsymbol{x}))  d\sigma(\boldsymbol{x}),
\qquad i=1,2.
\label{eq:current-i-th-ohmic}
\end{equation}
Following charge conservation, the currents in \eqref{eq:current-i-th-ohmic} satisfy:
\begin{equation}
\sum_{i=1}^2 j_{D_i}=0 \qquad \Rightarrow \qquad  j_{D_1}=- j_{D_2}=:i_D.
\label{eq:definition_iD}
\end{equation}

The voltage meter is simplified as a circuit with resistance $\mathcal{R}$, depicted in \Cref{fig:scheme-circ}. The potential difference between nodes $u_{D_1}$ and $u_{D_2}$ appearing in \eqref{eq:Dirichlet-data}, is given by 
\begin{equation}
u_{D_2} 
- u_{D_1} 
=\mathcal{R}\, i_D(u_{D_2}),
\label{eq:network-potentials}
\end{equation}
where $i_D$ is defined in \eqref{eq:definition_iD}. 
By assuming that one node of an electric circuit is grounded, we can freely assign $u_{D_1} = u_{\text{ref}} = 0$, resulting in the simplification of \eqref{eq:network-potentials} to:
\begin{equation}
u_{D_2} 
=\mathcal{R}\, i_D(u_{D_2}).
\label{eq:network-potentials-2}
\end{equation}

Equation \eqref{eq:network-potentials-2} is an implicit equation for $u_{D_2}$, given that $i_D$ depends implicitly on $u_{D_2}$ via the van Roosbroeck system \eqref{eq:vR-model}. After obtaining a solution $(\psi, \varphi_n, \varphi_p)$ to the van Roosbroeck system, where $u_{D_2}$ influences the Dirichlet boundary condition \eqref{eq:Dirichlet-data}, we compute the current $i_D$ using \eqref{eq:current-i-th-ohmic}. Further insights into this coupling can be found in \cite{ABGT} and \cite{Ali2010}. To simplify the notation, we will use $u_D$ instead of $u_{D_2}$.

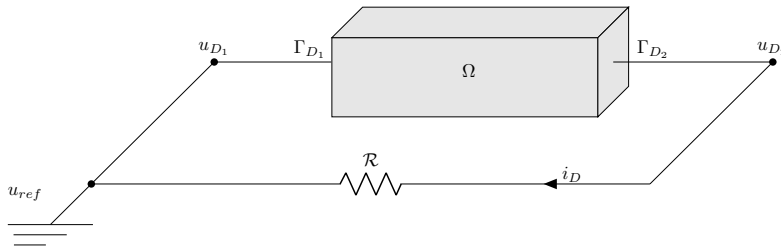
\begin{figure}

    \begin{center}
    \scalebox{0.7}{
\begin{tikzpicture}
    \coordinate (A) at (3,0,0);
    \coordinate (B) at (8,0,0);
    \coordinate (C) at (3,1.5,0);
    \coordinate (D) at (8,1.5,0);
    \coordinate (E) at (3,0,1.5);
    \coordinate (F) at (8,0,1.5);
    \coordinate (G) at (3,1.5,1.5);
    \coordinate (H) at (8,1.5,1.5);
    \coordinate (OM) at (5,0.3,0);
    \draw[thin,fill=gray!20] (E) -- (F) -- (H) -- (G) -- cycle;
    \draw[thin,fill=gray!20] (H) -- (G) -- (C) -- (D) -- cycle;
    \draw[thin,fill=gray!20] (H) -- (D) -- (B) -- (F) -- cycle;
    \coordinate (Center1) at (12,3,0.5);
    \node at (OM) {$\Omega$};
    \draw(8, 0.75, 0.75) -- (11, 0.75, 0.75);
    \draw(2.7, 0.75, 0.75) -- (0.5, 0.75, 0.75);
    \draw(0.5, 0.75, 0.75) -- (0.5, 0.75, 8.75);

    \draw(-0.3, 0.75, 8.75) -- (1.3, 0.75, 8.75);
    \draw(0, 0.75, 9.25) -- (1, 0.75, 9.25);
    \draw(0.2, 0.75, 9.75) -- (0.8, 0.75, 9.75);
   
     \draw(11, 0.75, 0.75) -- (11, 0.75, 6.75);

    \draw (0.5, 0.75, 6.75) to[R, l=$\mathcal{R}$] (11, 0.75, 6.75);
   \fill[black] (9, 0.75, 6.75) -- (9.25, 0.65, 6.75) -- (9.25, 0.85, 6.75) -- cycle;
   
\node at (9.25, 0.65, 6) {$i_D$};
\node at (8.75,1,0.75)  {$\Gamma_{D_2}$};
\node at (2.3,1,0.75)  {$\Gamma_{D_1}$};
\node at (0.5,1,0.75)  {$u_{D_1}$};
\node at (11,1,0.75)  {$u_{D_2}$};
\node at (-0.3,1,7.9)  {$u_{ref}$};
    \fill (0.5,0.75,0.75) circle (2pt);
\fill (11,0.75,0.75) circle (2pt);
\fill (0.5, 0.75, 6.75) circle (2pt);

\end{tikzpicture}}
\caption{A photo-sensitive semiconductor crystal $\Omega$ coupled to a voltmeter with  resistance $\mathcal{R}$.
}
\label{fig:scheme-circ}
\end{center}

\end{figure}

The total recombination rate, denoted as $R$, includes three mechanisms:
\begin{equation}
R= R_{\si{dir}} + R_{\si{Aug}} + R_{\si{SRH}},
\label{eq:reco}
\end{equation}
where $R_{\si{dir}}$, $R_{\si{Aug}}$, and $R_{\si{SRH}}$ represent direct recombination, Auger recombination, and Shockley-Read-Hall (SRH) recombination, respectively, defined as:
\begin{equation}
\begin{split}
 R_{\si{dir}} &=C_d(np-n_{\text{i}}^2),
\\
  R_{\si{Aug}} &=C_{\n}n(np-n_{\text{i}}^2)+C_{\p}p(np-n_{\text{i}}^2),
\\
  R_{\si{SRH}} &=\frac{np-n_{\text{i}}^2}{\tau_p(n+n_T)+\tau_n(p+p_T)},
  \end{split} 
\end{equation}
where $C_d$, $C_n$ and $C_p$ are material-related constants. 
The SRH recombination process, which involves the trapping and release of charge carriers in semiconductors, differs from the Auger recombination process due to unintentional inclusion of elements during fabrication. This unintentional inclusion affects the properties of the semiconductor material, such as the lifetimes of charge carriers ($\tau_n$ for electrons, $\tau_p$ for holes) and the reference densities of charge carriers ($n_T$ for electrons, $p_T$ for holes). For later use we define the function 
\begin{equation}
  r(n,p):=C_d+C_{\n}n+C_{\p}p+\frac{1}{\tau_p(n+n_T)+\tau_n(p+p_T)}
  \label{eq:def-r}
\end{equation}
so that $R=r(n,p)(np-n_{\text{i}}^2)$. 

The electromagnetic source  is modeled by the generation term $G(\mathbf{x}; \mathbf{x}_0)$. When the laser hits the crystal at the point $\mathbf{x}_0 := (x_0, y_0, z_0)^T$, some photons are \textit{impinged} and create electron-hole pairs, resulting in a  generation rate defined as follows
\begin{equation}
G(\mathbf{x}; \mathbf x_0)=\kappa S(\mathbf{x}-\mathbf{x}_0),
\label{eq:generation-rate}
\end{equation}
where $S(\mathbf{x})$ is the shape function of the laser (normalized by $\int_{\mathbb{R}^3}
S(\mathbf{x}){d}\mathbf{x}=1$), while $\kappa$ is a constant given by
$\kappa := \frac{P\lambda_L}{h c}(1-\rho)$, where $c$ is the velocity of light in vacuum. 
Here, $P$ is the laser power, $\lambda_{{L}}$ is the wave length of the laser, $h$ is the Planck constant, and $\rho$ is the reflectivity rate of the crystal.

We assume that area of influence of the electromagnetic source decays exponentially fast from the incident point $\mathbf{x}_0$. In particular, we take a laser profile function $S$ defined as
\begin{equation}
    S(\mathbf{x}) := \frac{1}{2\pi\sigma_L^2 d_A}
    \exp\left[-\frac12 \left( \frac{x}{\sigma_L}\right)^2 \right]
    \exp\left[-\frac12 \left( \frac{y}{\sigma_L}\right)^2 \right]
    \exp\left[-\frac{|z|}{d_{A}}\right].
\label{eq:shapeS}
\end{equation}
Here $\sigma_L$ is the laser spot radius, while $d_{A}$ is the penetration depth
(or the reciprocal of the absorption coefficient), which  heavily depends on the
laser wave length.

\begin{figure}[h!]
	\centering  
  \hfill
  \includegraphics[width=0.49\columnwidth]{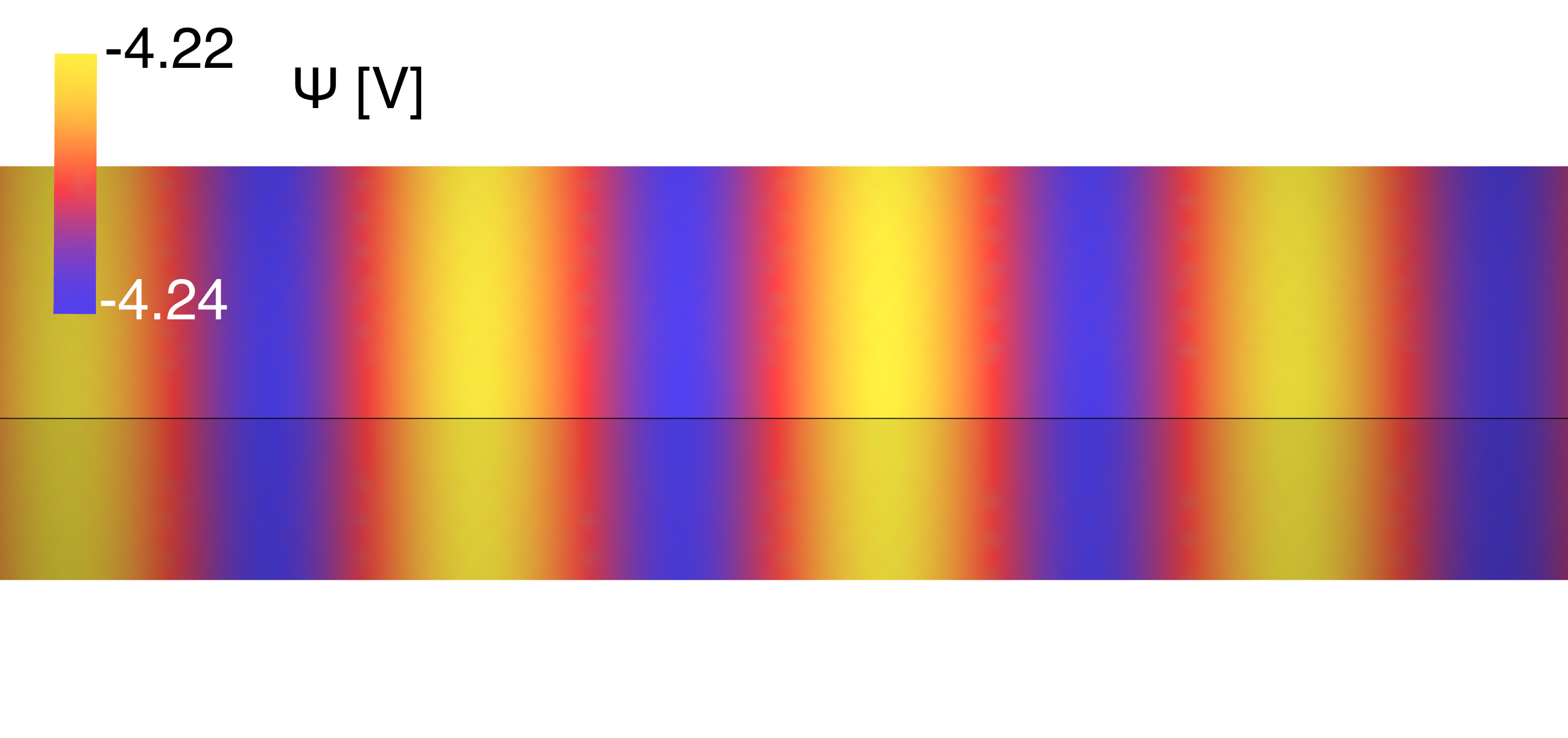}
  \hfill
  \includegraphics[width=0.49\columnwidth]{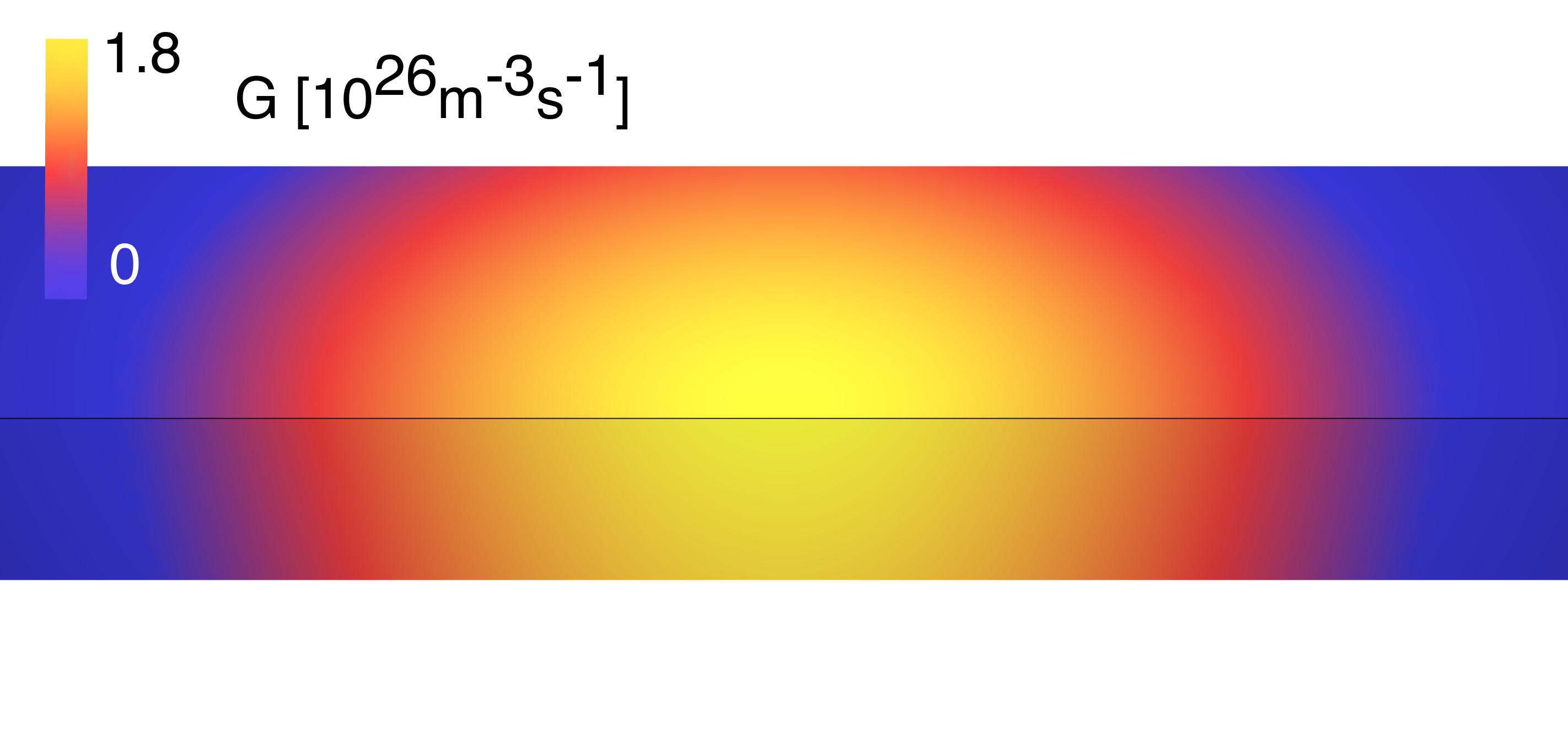}
  \hfill
  \\
  \hfill
  \includegraphics[width=0.49\columnwidth]{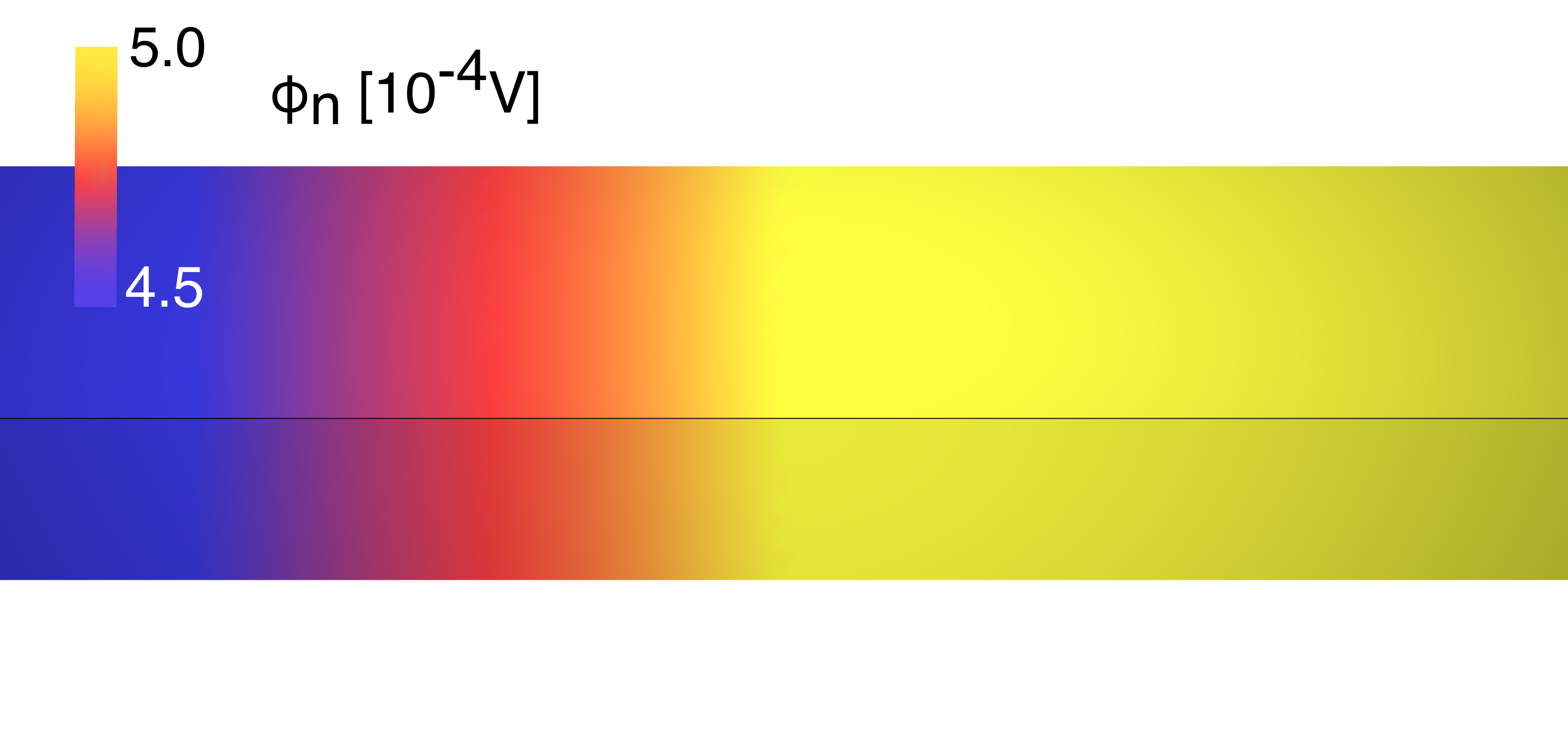}
  \hfill
  \includegraphics[width=0.49\columnwidth]{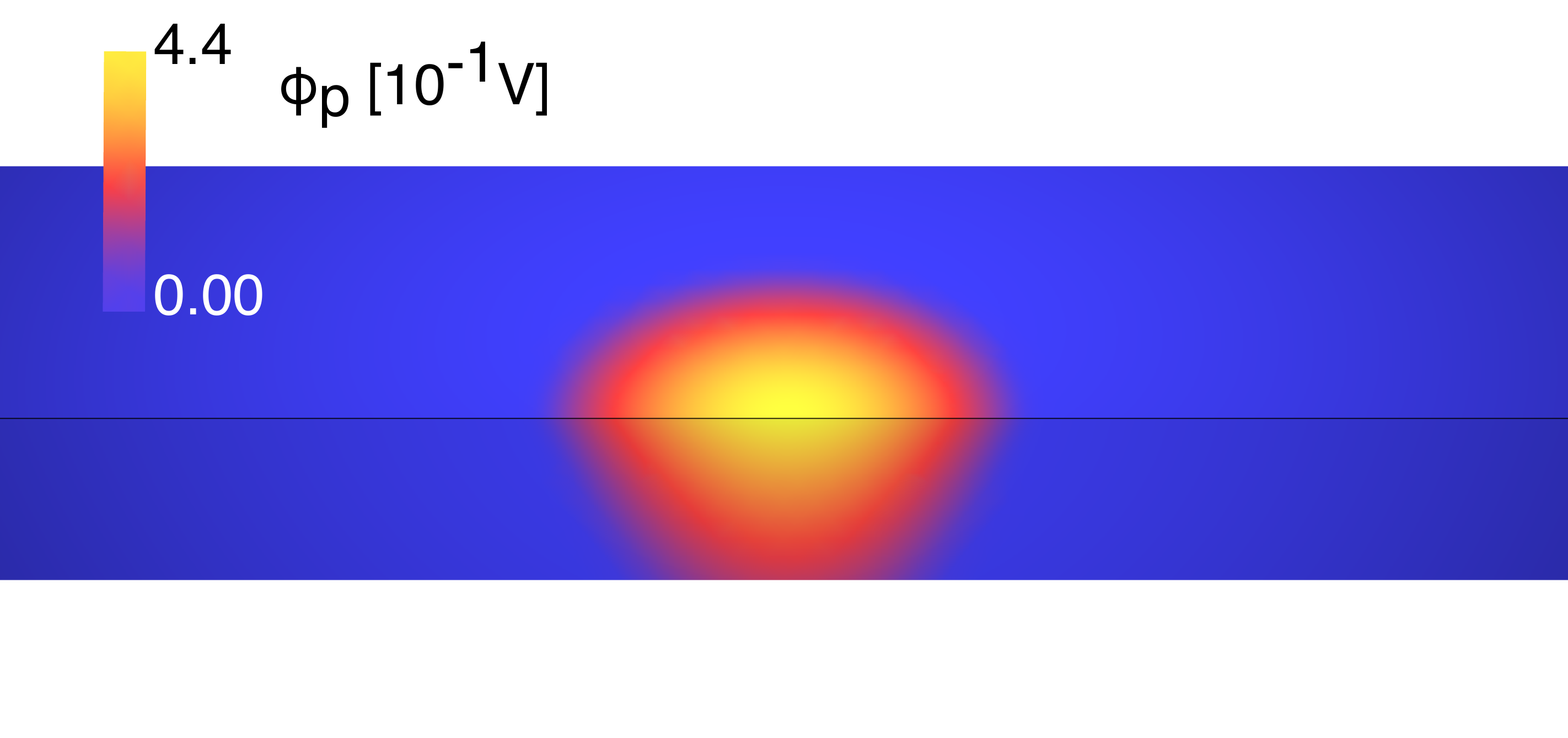}
    \\
  \hfill
  \includegraphics[width=0.49\columnwidth]{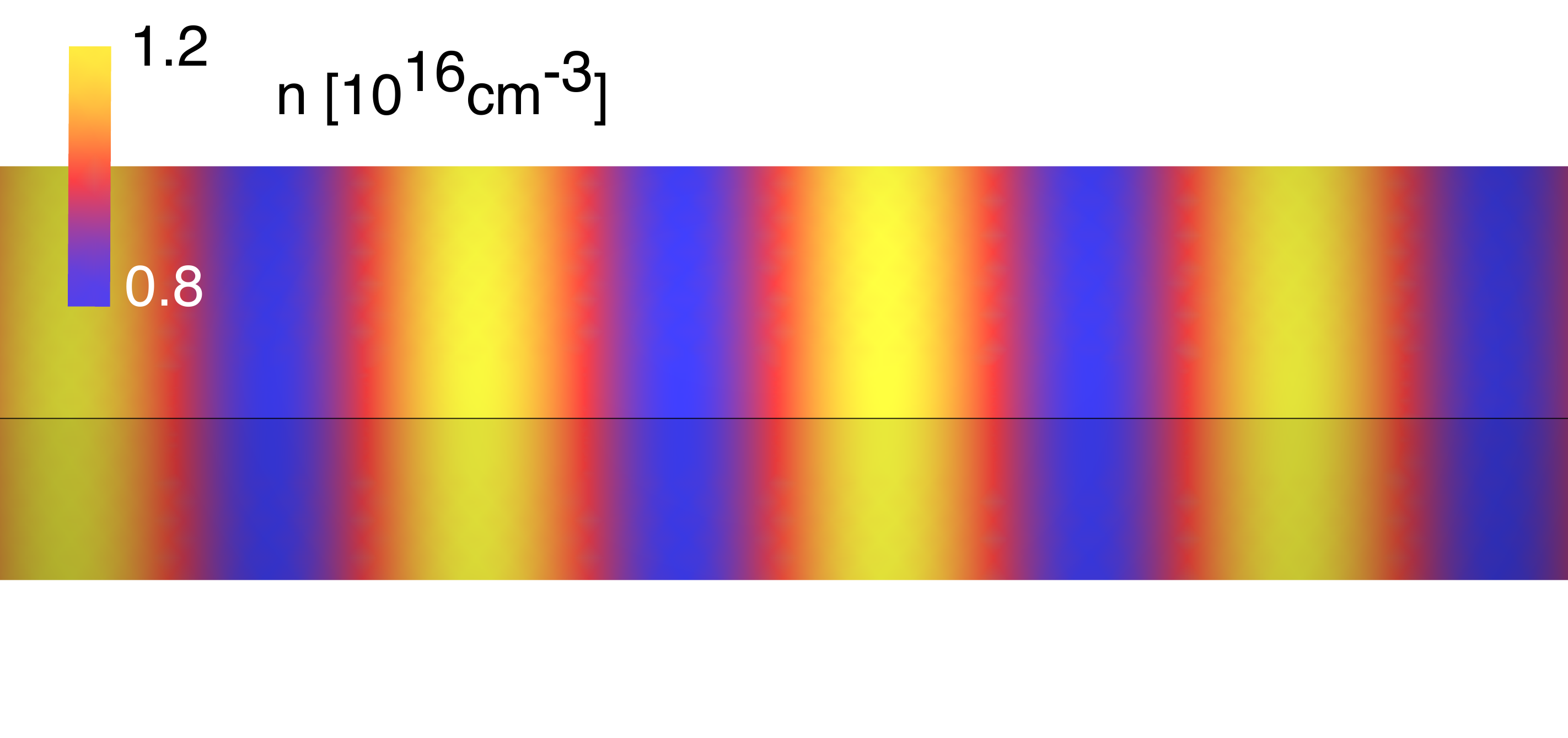}
  \hfill
  \includegraphics[width=0.49\columnwidth]{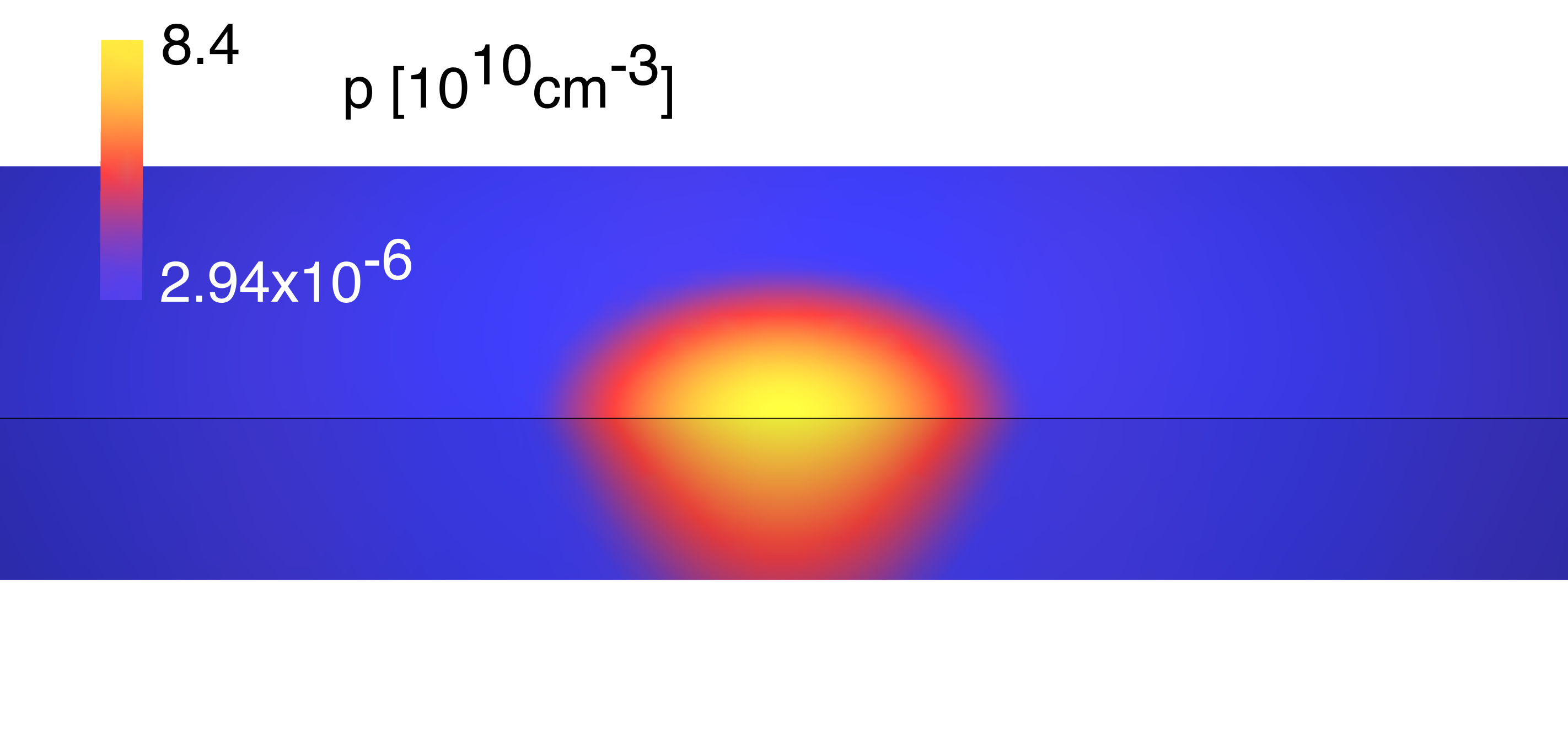}
\caption{For a silicon crystal, we show the electrostatic potential $\psi$ (upper left panel), the resulting generation rate $G$ (upper right panel), the quasi Fermi potentials $\varphi_n$,$\varphi_p$ for electrons and holes (middle panels). The related charge carrier densities $n,p$ are shown in the bottom panels.}
\label{fig:3D_simulation-figures_Si} 
\end{figure}

\subsection{Numerical simulations}
For the above model, we performed  numerical simulation with \texttt{ddfermi} \cite{ddfermi} of a donor-doped semiconductor crystal irradiated by a laser beam, both for silicon (see Figure \ref{fig:3D_simulation-figures_Si}) and gallium arsenide (see Figure \ref{fig:3D_simulation-figures_GaAs}). We set the laser power $P=\SI{2}{mW}$, laser spot position $(x_0,y_0)^T = (\SI{0}{mm},\SI{0}{mm})^T$, average doping concentration $N_{D0}=\SI{1E16}{cm^{-3}}$ and $C(x,y,z) =  N_{D0} \left(1+0.2\sin\left(2\pi\frac{x}{\SI{100}{\micro\meter}}\right)\right)$.  
The additional parameters for both simulations can be found in \ref{sec:parameters}. The results show that the presence of the beam, modeled by the generation term $G$, greatly affects the minority charge carriers $p$, while the majority charge carriers $n$ are not significantly impacted. As a matter of fact, the sinusoidal pattern of the doping is reflected in the electron densitiy and the electrostatic potential alone.


\begin{figure}[h!]
	\centering  
  \hfill
  \includegraphics[width=0.49\columnwidth]{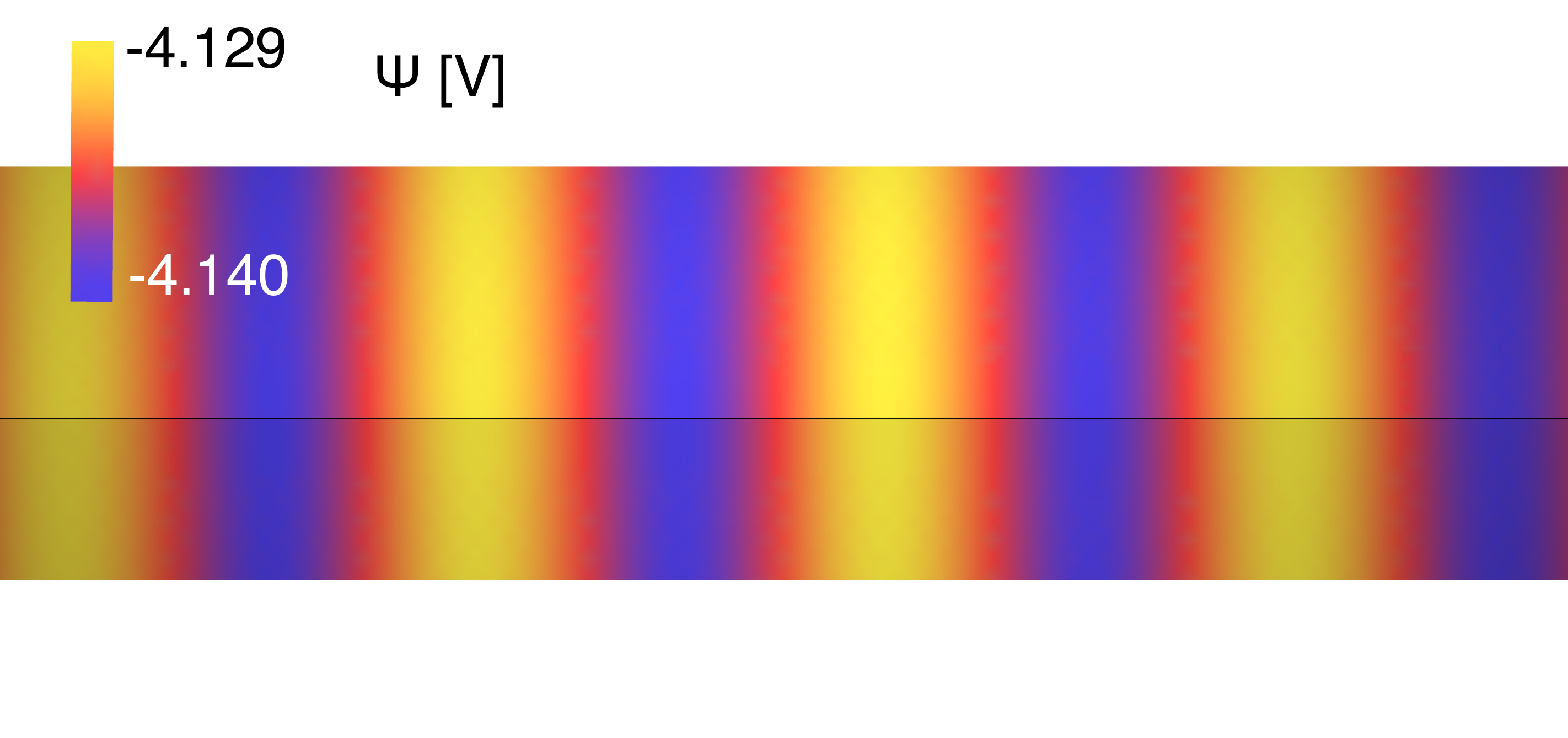}
  \hfill
  \includegraphics[width=0.49\columnwidth]{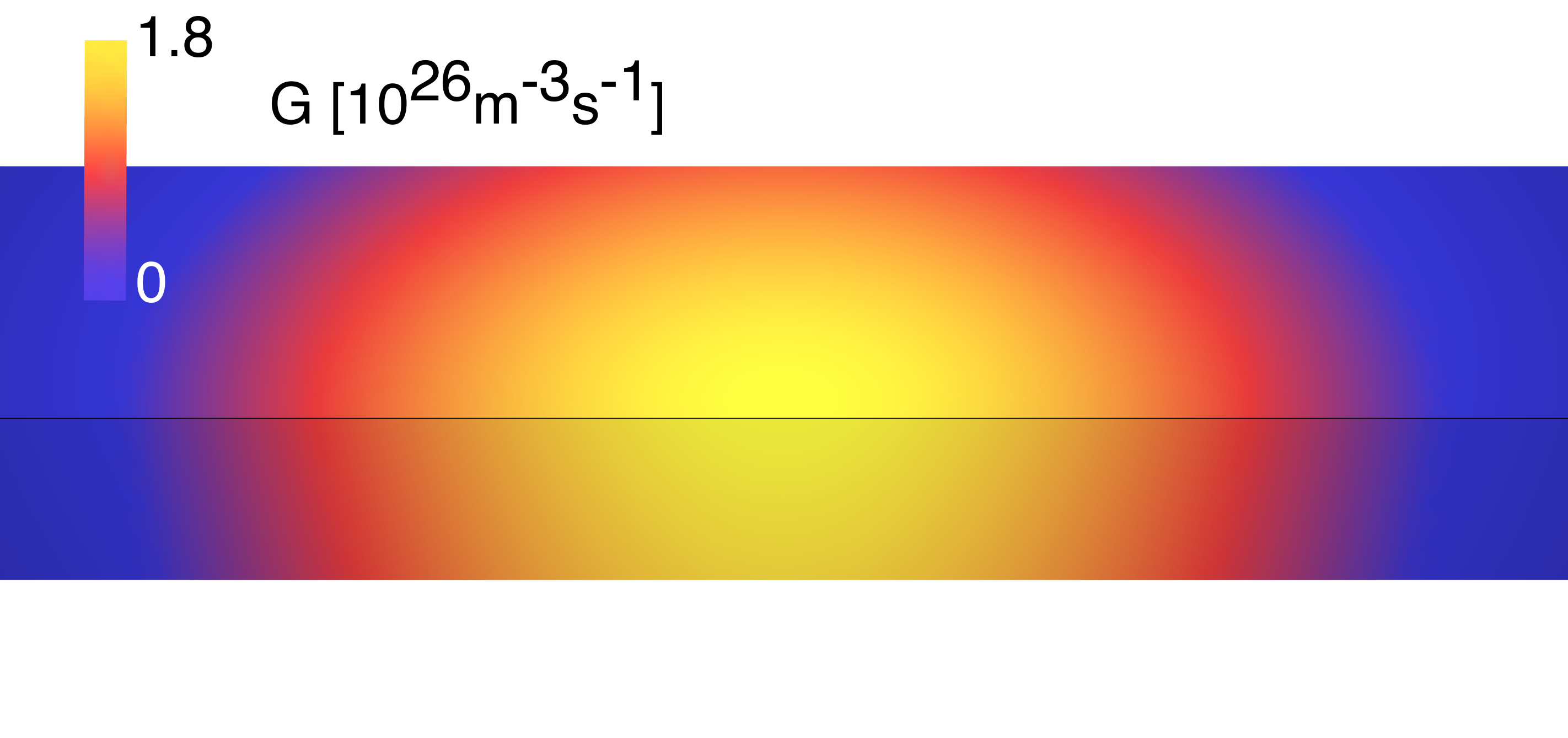}
  \hfill
  \\
  \hfill
  \includegraphics[width=0.49\columnwidth]{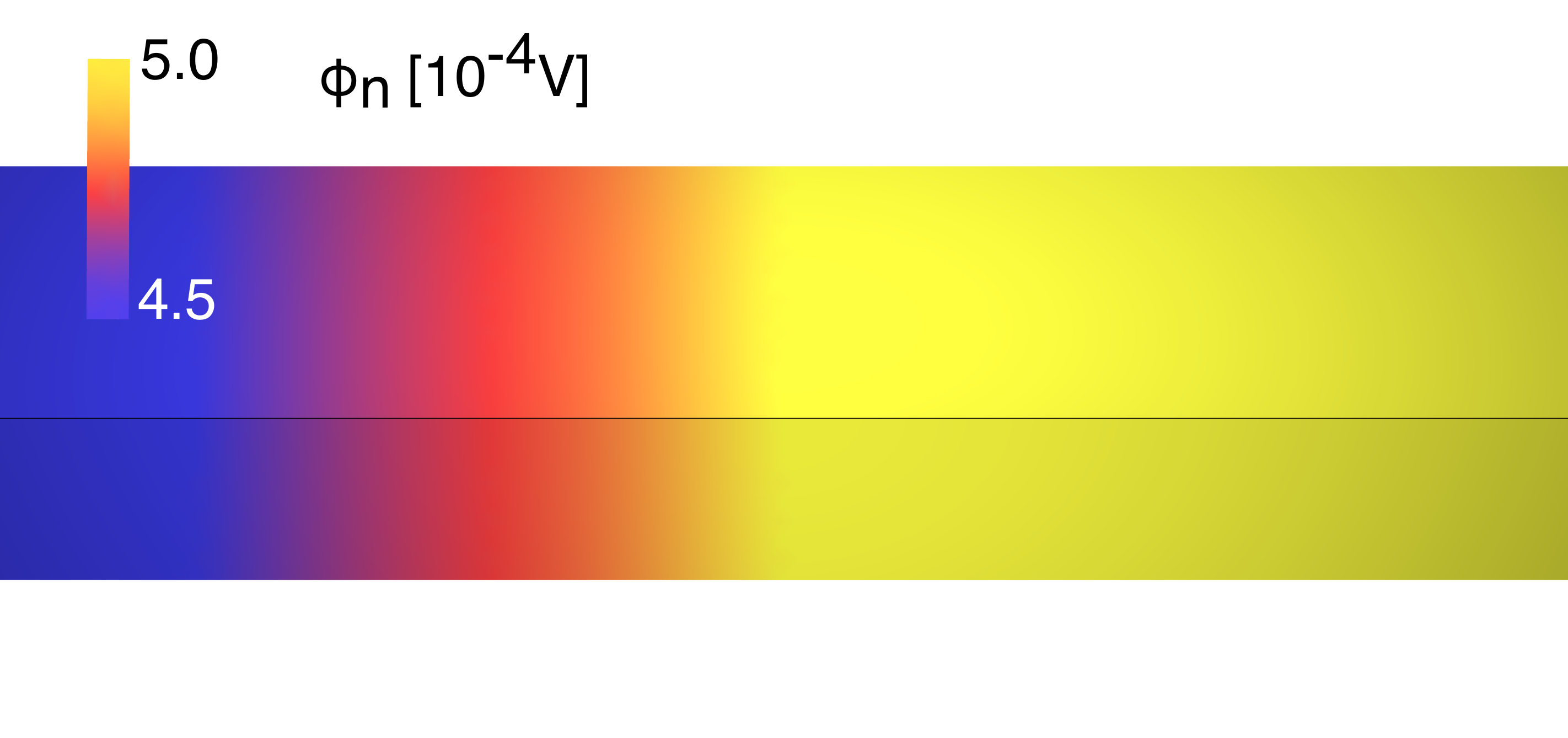}
  \hfill
  \includegraphics[width=0.49\columnwidth]{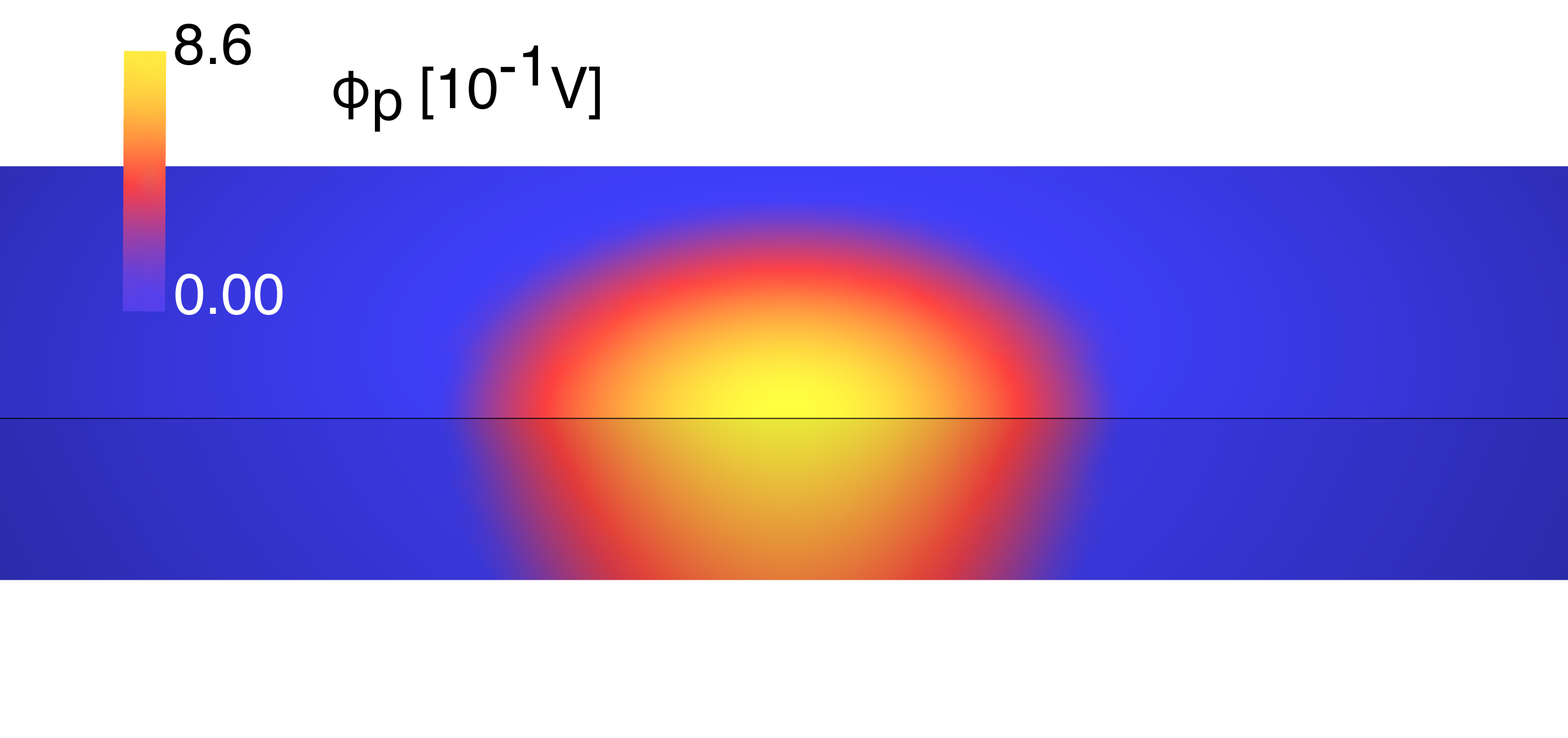}
    \\
  \hfill
  \includegraphics[width=0.49\columnwidth]{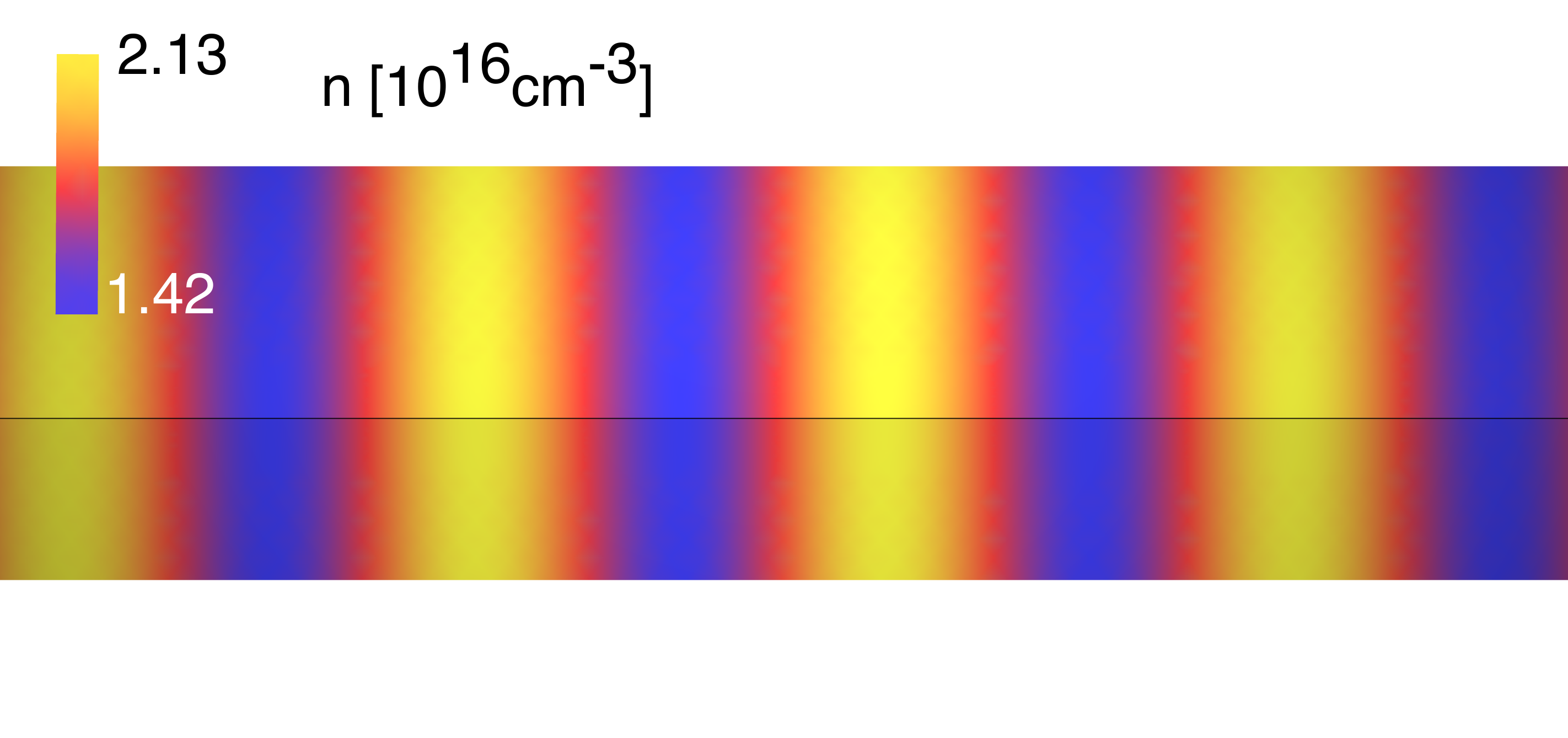}
  \hfill
  \includegraphics[width=0.49\columnwidth]{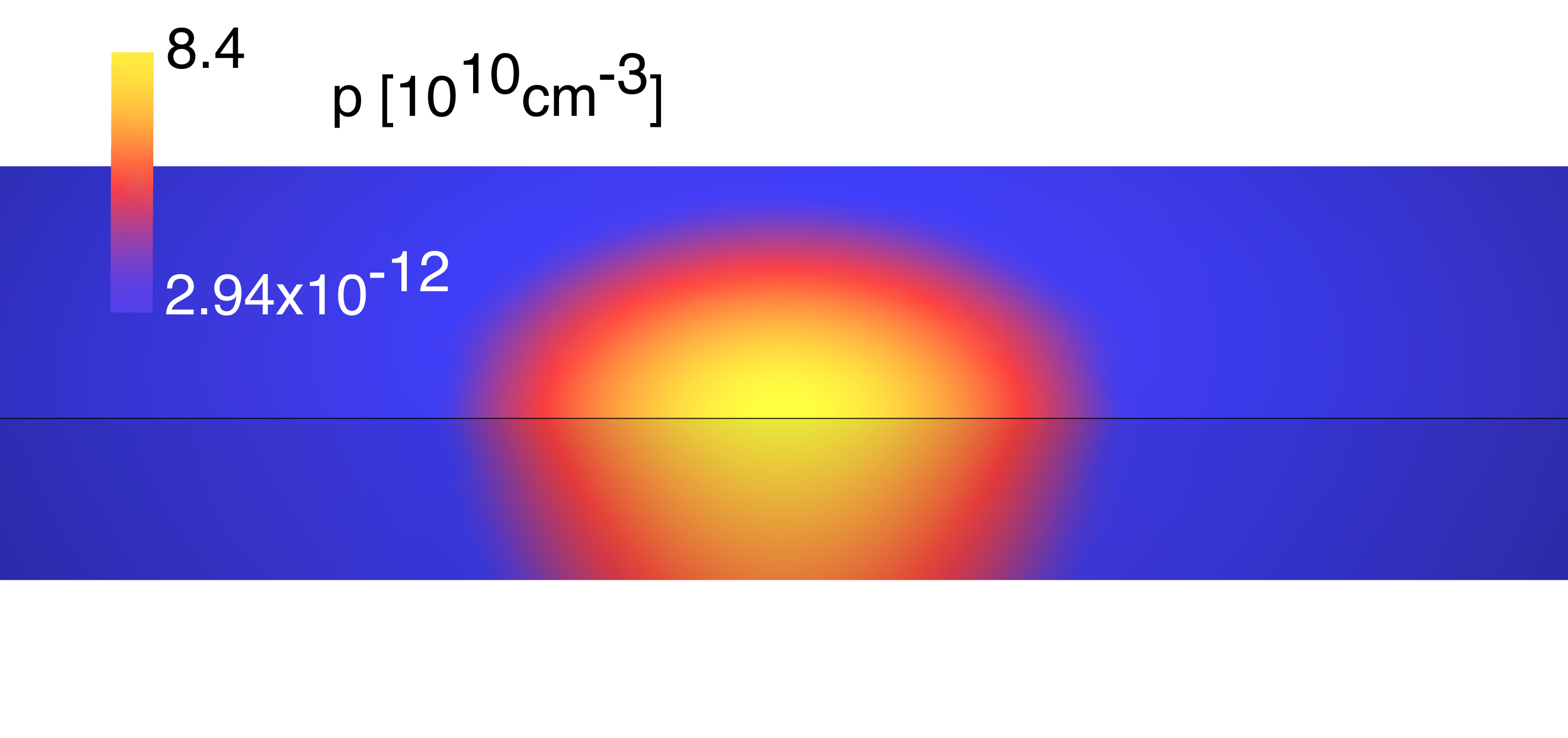}
\caption{For a gallium arsenide crystal, we show the electrostatic potential $\psi$ (upper left panel), the resulting generation rate $G$ (upper right panel), the quasi Fermi potentials $\varphi_n$,$\varphi_p$ for electrons and holes (middle panels). The related charge carrier densities $n,p$ are shown in the bottom panels.}
\label{fig:3D_simulation-figures_GaAs} 
\end{figure}

These numerical findings are the main motivation for a unipolar scaling of the full model, also encompassing the electric network variables, which will be detailed in the following section.

\section{Scaling of the model for unipolar devices}
\label{sec:scaling}

We consider a scaling of the model \eqref{eq:vR-model} which is suitable for unipolar devices. In this case, the density of the majority charge carriers is of the same order as the doping profile, while the density of the minority charge carriers is much smaller. Nevertheless, both charge carriers contribute to the total current, especially when the current is the result of an impinging laser beam.

In the following we consider an $n$-doped device, so that $C(\boldsymbol{x})>0$, and electrons are the majority charge carriers.
We write
\beaa
&&
\varphi_n=\bar{\psi}\hat{\varphi}_{n}-\varphi_0, \quad \varphi_p=\bar{\psi}\hat{\varphi}_{p}-\varphi_0, \quad \psi=\bar{\psi}\hat{\psi},
\quad \boldsymbol{x}=\bar{x}\hat{\boldsymbol{x}}, 
\\ && 
\quad \mu_n=\bar{\mu}\hat{\mu}_n, \quad \mu_p=\bar{\mu}\hat{\mu}_p, \quad \varepsilon({\boldsymbol{x}})=\bar{\varepsilon}\hat{\varepsilon}(\hat{\boldsymbol{x}}), \quad
C(\boldsymbol{x})=\bar{C}\hat{C}(\hat{\boldsymbol{x}}), \quad
\eeaa 
where $\bar{\psi}$, $\varphi_0$, $\bar{x}$, $\bar{\mu}$, $\bar{\varepsilon}$, $\bar{C}$  are reference values, while $\hat{\varphi}_{n}$, $\hat{\varphi}_{p}$, $\hat{\psi}$, $\hat{\boldsymbol{x}}$, $\hat{\mu}_n$, $\hat{\mu}_p$, $\hat{\varepsilon}$, $\hat{C}$,  are scaled quantities. We choose
\beaa
&&
\bar{C}=\sup_{\boldsymbol{x}\in\Omega} C(\boldsymbol{x}),
\quad 
\bar{\psi}=V_{\mathrm{th}}:=\frac{k_{B}T}{q},
\quad \varphi_0=\frac{E_\cc}{q}-V_{\mathrm{th}}\log\frac{N_\cc}{\bar{C}},
\quad \bar{x}=\mathrm{diam}\, \Omega,
\\&&
 \tau=\frac{\bar{x}^2}{\bar{\mu}V_{\mathrm{th}}},\quad 
\bar{\mu}=\sup_{\boldsymbol{x}\in\Omega}\max\{\mu_n(\boldsymbol{x}),\mu_p(\boldsymbol{x})\}, 
\quad 
\bar{\varepsilon}=\sup_{\boldsymbol{x}\in\Omega} \varepsilon(\boldsymbol{x}), 
\eeaa
We  notice explicitely that from \eqref{eq:dens-pot}, the scaling implies
\bea \hat{n}(\hat{\varphi}_n,\hat{\varphi}_p)=\exp({\hat{\psi}-\hat{\varphi}_n})
, \quad
\hat{p}(\hat{\varphi}_n,\hat{\varphi}_p)=\frac{n_\mathrm{i}^2}{\bar{C}^2} \exp({\hat{\varphi}_p-\hat{\psi}}).
\eea 
We scale the generation term $G$ and the recombination term $R$ as follows  
\beaa 
G(\boldsymbol{x})=\bar{G}\hat{G}(\hat{\boldsymbol{x}}), \quad
R(n,p)=\bar{R}\hat{R}(\hat{n},\hat{p}),
\eeaa
with   
\beaa
&&
\bar{R}=\frac{\bar{C}}{\tau}, 
\quad 
\bar{G}= \frac{n_\mathrm{i}^2}{\bar{C}\tau},
\eeaa
and
\begin{equation}
    \begin{split}
\hat{R}(\hat{n},\hat{p}) &=\hat{r}(\hat{n},\hat{p})(\hat{n}\hat{p}-n_\mathrm{i}^2/\bar{C}^2) 
\\
  \hat{r}(\hat{n},\hat{p}) &=\hat{C}_d+\hat{C}_{\n}\hat{n}+\hat{C}_{\p}\hat{p} +\frac{1}{\hat{\tau}_p(\hat{n}+\frac{n_{\text{i}}}{\bar{C}}\hat{n}_T)+\hat{\tau}_n(\hat{p}+\frac{n_{\text{i}}}{\bar{C}}\hat{p}_T)}        
    \end{split}
    \label{eq:scaled-R}
\end{equation}
where  
\beaa
&&
\hat{C}_d=\tau\bar{C}C_d, 
\quad
\hat{C}_{\n}=\tau\bar{C}^2C_{\n},
\quad
\hat{C}_{\p}=\tau\bar{C}^2C_{\p},
\\&&
\hat{\tau}_n=\frac{\tau_n}{\tau},
\quad
\hat{\tau}_p=\frac{\tau_p}{\tau},
\quad
\hat{n}_{T}=\frac{n_T}{n_{\text{i}}},
\quad
\hat{p}_{T}=\frac{p_T}{n_{\text{i}}}.
\eeaa

The system \eqref{eq:vR-model} becomes
\begin{equation}
  \begin{cases}
-\frac{\bar{\varepsilon}V_{\mathrm{th}}}{q\bar{C}\bar{x}^2} \hat{\nabla} \cdot (\hat{\varepsilon} \hat{\nabla} \hat{\psi}) &= 
\hat{p}-\hat{n}+\hat{C}(\hat{\boldsymbol{x}})
\\[2ex]
- \hat{\nabla}\cdot ( \hat{\mu}_\n  \hat{n} \hat{\nabla} \hat{\varphi}_\n ) &= \hat{R}(\hat{n},\hat{p}) -\frac{n_\mathrm{i}^2}{\bar{C}^2}\hat{G}(\hat{\boldsymbol{x}}), \qquad
\\[2ex]
-\hat{\nabla}\cdot (\hat{\mu}_\p  \hat{p} \hat{\nabla} \hat{\varphi}_{\p}) &= \frac{n_\mathrm{i}^2}{\bar{C}^2}\hat{G}(\hat{\boldsymbol{x}})-\hat{R}(\hat{n},\hat{p}), \qquad
\end{cases}
\label{eq:vR-scaled-model}
\end{equation}
The current densities scale according to: 
\begin{equation}
  \boldsymbol{J}_\n = q \frac{\bar{\mu }\bar{C}V_{\mathrm{th}}}{\bar{x}}\hat{\boldsymbol{J}_\n}, \qquad   \boldsymbol{J}_\p = q \frac{\bar{\mu }{\bar{C}}V_{\mathrm{th}}}{\bar{x}}\hat{\boldsymbol{J}_\p}.
  \label{eq:currents-hat}
\end{equation}
The boundary conditions become:
\begin{equation}
  \frac{\partial \hat{\psi}}{\partial \boldsymbol{\nu}}=
  \frac{\partial \hat{\varphi}_\n}{\partial \boldsymbol{\nu}}=
  \frac{\partial \hat{\varphi}_\p}{\partial \boldsymbol{\nu}}=0,  \quad \text{on }\Gamma_N,
  \label{eq:neumann-data-hat}
  \end{equation}
  \begin{equation}
    \begin{split}&
      \hat{\psi} - \hat{\psi}_0= \hat{\varphi}_\n -\hat{\varphi}_0= \hat{\varphi}_\p -\hat{\varphi}_0= 0\quad \text{on }\quad\Gamma_{D_1}, 
    \\&
    \hat{\psi} - \hat{\psi}_0= \hat{\varphi}_\n -\hat{\varphi}_0= \hat{\varphi}_\p -\hat{\varphi}_0= \hat{u}_{D}  \quad \text{on }\quad\Gamma_{D_2}, 
    \label{eq:Dirichlet-data-hat}
    \end{split}
    \end{equation}
    where $\psi_0=V_\mathrm{th}\hat{\psi}_0$, $\varphi_0=V_\mathrm{th}\hat{\varphi}_0$, $u_{D}  = \bar{ u}_{D} \hat{ u}_{D} = V_{\mathrm{th}}\hat{ u}_{D}$. 
    Using \eqref{eq:current-i-th-ohmic}, \eqref{eq:definition_iD}  and \eqref{eq:currents-hat} we have 
    \begin{equation}
      \begin{split}
      i_D(u_{D})  &= \displaystyle\int_{{\Gamma}_{D_2}}  {\boldsymbol{\nu}} \cdot ({\boldsymbol{J}_\n} (\boldsymbol{x})+{\boldsymbol{J}_\p} (\boldsymbol{x}))  d{\sigma}(\boldsymbol{x})
\\
 &=   \bar{i}_D \int_{\hat{\Gamma}_{D_2}}  {\boldsymbol{\nu}} \cdot ( \hat{\boldsymbol{J}_\n}+  \hat{\boldsymbol{J}_\p} ) d\hat{\sigma}(\boldsymbol{x})
 =\bar{i}_D\hat{i}_D(\hat{u}_{D}) , 
    \label{eq:current-i-th-ohmic-hat}
  \end{split}
\end{equation}
with
 $\bar{i}_D= q  \bar{\mu }\bar{C} V_{\mathrm{th}} {\bar{x}}$, and
 \[
 \hat{i}_D(\hat{u}_{D})=\int_{\hat{\Gamma}_{D_2}}  {\boldsymbol{\nu}} \cdot ( \hat{\boldsymbol{J}_\n}+\hat{\boldsymbol{J}_\p} ) d\hat{\sigma}(\boldsymbol{x}).
 \]
The coupling equation \eqref{eq:network-potentials-2} becomes 
  \begin{equation}
    \begin{split}
      \hat{ u}_{D}   =  
      \hat{\mathcal{R}}\, \hat{i}_D(\hat{u}_{D}), 
      \quad \text{with } \hat{\mathcal{R}}=q \bar{\mu }\bar{C} {\bar{x}}  {\mathcal{R}}. 
      \label{eq:network-potentials-2-hat}
    \end{split}
    \end{equation}

Omitting the $\hat{ }$ and introducing the nondimensional parameters
\bea
\label{eq:scaling_parameters}
&&
\lambda=\sqrt{\frac{\bar{\varepsilon}V_{\mathrm{th}}}{q\bar{C}\bar{x}^2}},
\quad
\delta=\frac{n_{\text{i}}}{\bar{C}},
\eea
we obtain 
\begin{equation}\begin{cases}
-\lambda^2 \nabla \cdot (\varepsilon \nabla \psi) &= 
\delta^2 e^{\varphi_p-\psi} - e^{\psi-\varphi_n}  +C(\boldsymbol{x}),
\\[2ex]
- \nabla\cdot ( \mu_\n  e^{\psi-\varphi_n} \nabla \varphi_\n ) &= \delta^2r_{\delta}(n,p) ( e^{\varphi_p-\varphi_n}-1) -\delta^2 G(\boldsymbol{x}), \qquad
\\[2ex]
- \nabla\cdot (\mu_\p   e^{\varphi_p-\psi} \nabla \varphi_{\p}) &=   G(\boldsymbol{x})-r_{\delta}(n,p) ( e^{\varphi_p-\varphi_n}-1),  
\end{cases}
  \label{eq:vR-scaled-model_fin}
\end{equation}
in which we have used $n=e^{\psi-\varphi_n}$, $p=\delta^2 e^{\varphi_p-\psi}$ and set
\be
r_{\delta}(n,p)
={C}_d+{C}_{\n}n+{C}_{\p} p +\frac{1}{{\tau}_p(n+\delta n_T)+{\tau}_n(p+\delta {p}_T)}.
\label{r_delta}
\ee
The boundary and coupling conditions are 
\begin{equation}
  \frac{\partial {\psi}}{\partial \boldsymbol{\nu}}=
  \frac{\partial {\varphi}_\n}{\partial \boldsymbol{\nu}}=
  \frac{\partial {\varphi}_\p}{\partial \boldsymbol{\nu}}=0,  \quad \text{on }\Gamma_N,
  \label{eq:neumann-data-hat-limit-app}
  \end{equation}
  \begin{equation}
    \begin{split}&
      {\psi} - {\psi}_0= {\varphi}_\n - {\varphi}_0= {\varphi}_\p-{\varphi}_0 = 0\quad \text{on }\quad\Gamma_{D_1}, 
    \\&
    {\psi} - {\psi}_0= {\varphi}_\n - {\varphi}_0= {\varphi}_\p - {\varphi}_0= {u}_{D}  \quad \text{on }\quad\Gamma_{D_2}, 
    \label{eq:Dirichlet-data-hat-limit-app}
    \end{split}
    \end{equation}
    \be
{u}_D   =  
      {\mathcal{R}} i_D(u_D),
\qquad 
  {i}_D({u}_{D})=\displaystyle\int_{{\Gamma}_{D_2}}  {\boldsymbol{\nu}} \cdot ( {\boldsymbol{J}_\n}+{\boldsymbol{J}_\p} ) d{\sigma}(\boldsymbol{x}).
  \label{eq:coupling}
  \ee

We compute the dimensionless parameters in \eqref{eq:scaling_parameters} for two standard seminconductors, namely silicon and gallium arsenide. The corresponding parameters can be found in \ref{sec:parameters}.
The laser power $P$ is chosen for both cases to be 20mW and the conduction band-edge energy $E_c = k_B T\log(N_c/\bar{C})$.
For a silicon crystal, as described in \cite{Farrell2021,Luedge1997, Kayser2020b}, the parameters in \eqref{eq:scaling_parameters} become
\bea
\label{eq:scaling_parameters2}
&&
\lambda\approx 1.249382\times 10^{-5},
\quad
\delta\approx 5.528936\times 10^{-7}.
\eea
%
%
For a gallium arsenide crystal, as described in \cite{Kayser2020d}, the parameters in \eqref{eq:scaling_parameters} become
\bea
\label{eq:scaling_parameters3}
&&
\lambda\approx  1.306319 \times 10^{-6},
\quad
\delta\approx 2.154036 \times 10^{-12}.
\eea

These calculations show that for two standard semiconductor crystals the parameter $\delta$ is indeed a small parameter, and it is appropriate to perfom an asymptotic analysis for $\delta$ tending to zero.
In this limit, the system \eqref{eq:vR-scaled-model_fin} becomes 
\begin{equation}\begin{cases}
  &-\lambda^2 \nabla \cdot (\varepsilon \nabla \psi) = 
    - e^{\psi-\varphi_n}  +C(\boldsymbol{x}),
\\[2ex]
 & - \nabla\cdot ( \mu_\n  e^{\psi-\varphi_n} \nabla \varphi_\n ) = 0
\\[2ex]
 & - \nabla\cdot (\mu_\p   e^{\varphi_p-\psi} \nabla \varphi_{\p}) =   G(\boldsymbol{x})-{r}_0(e^{\psi-\varphi_n}) ( e^{\varphi_p-\varphi_n}-1), 
  \\[1ex]&\displaystyle  \frac{\partial {\psi}}{\partial \boldsymbol{\nu}}=
  \frac{\partial {\varphi}_\n}{\partial \boldsymbol{\nu}}=
  \frac{\partial {\varphi}_\p}{\partial \boldsymbol{\nu}}=0,  \quad \text{on }\Gamma_N,
  \\&{\psi} - {\psi}_0= {\varphi}_\n - {\varphi}_0= {\varphi}_\p - {\varphi}_0= 0\quad \text{on }\quad\Gamma_{D_1}, 
  \\&
  {\psi} - {\psi}_0= {\varphi}_\n - {\varphi}_0= {\varphi}_\p - {\varphi}_0= {u}_{D}  \quad \text{on }\quad\Gamma_{D_2}, 
  \\&
  {u}_D   =  
  {\mathcal{R}} i_D(u_D),
\qquad 
{i}_D({u}_{D})=\displaystyle\int_{{\Gamma}_{D_2}}  {\boldsymbol{\nu}} \cdot {\boldsymbol{J}_\n}  d{\sigma}(\boldsymbol{x}).
  \end{cases}
    \label{eq:vR-scaled-model_fin-delta=0}
  \end{equation}
Here we have used the notation $r_0$ to indicate that we put  $\delta = 0$ in \eqref{r_delta} so that 
  \be 
  r_{0}(n)=C_d+C_{\n}n +\frac{1}{\tau_p n}.
  \label{eq:r-0}
  \ee
  We notice that the hole current for $\delta=0$ vanishes since
  \be
  \boldsymbol{J}_p = -\delta^2 \mu_p e^{\varphi_p-\psi}\nabla \varphi_p. 
  \ee
 and from the second equation in \eqref{eq:vR-scaled-model_fin-delta=0} the electron current density $\boldsymbol{J}_n$ is constant, this implies that ${i}_D({u}_{D})=0$ which means there is no coupling between the circuit and the drift diffusion system at the lowest order. For this reason, we perform an aymptotic expansion in the next section to find leading-order models which capture the most important characteristics from \eqref{eq:vR-scaled-model_fin}.  


\section{Asymptotic expansion}
\label{sec:expansion} 
We consider the following asymptotic expansions:
\be
\psi=\sum_{k\ge0} \psi^{(k)}\delta^k, \quad 
\varphi_\n=\sum_{k\ge0} \varphi_\n^{(k)}\delta^k, \quad 
\varphi_\p=\sum_{k\ge0} \varphi_\p^{(k)}\delta^k, \quad 
u_D=\sum_{k\ge0} u_D^{(k)}\delta^k,
\label{exp1}
\ee
which in turns yield expansions for $n=\exp(\psi-\varphi_n)$, $p=\delta^2 \exp({\varphi_p-\psi})$,
$R=\delta^2 r_{\delta} \left(\frac{np}{\delta^2}-1\right) $.
\begin{lemma}
Assuming expansions \eqref{exp1}, we have
\be
\n(\psi, \varphi_n)=\sum_{k\ge0} \n^{(k)}\delta^k, \qquad 
\p(\psi, \varphi_p)=\delta^2\sum_{k\ge0} \p^{(k)}\delta^k, \qquad 
R=\delta^2 \sum_{k\ge0} R^{(k)}\delta^k,
\ee
with
\be
\begin{split}
n^{(0)}&= \exp(\psi^{(0)}-\varphi_n^{(0)}), \quad
n^{(1)}=n^{(0)}(\psi^{(1)}-\varphi_n^{(1)}), 
\\
n^{(k)}&=n^{(0)}(\psi^{(k)}-\varphi_n^{(k)}) +n^{(0)}
F^{(k)}
\left(\frac{n^{(1)}}{n^{(0)}},\dots,\frac{n^{(k-1)}}{n^{(0)}}\right), \quad k\ge 2,
\\
p^{(0)}&= \exp(\varphi_p^{(0)}-\psi^{(0)}), \quad
p^{(1)}=p^{(0)}(\varphi_\p^{(1)}-\psi^{(1)}), 
\\
p^{(k)}&=p^{(0)}(\varphi_\p^{(k)}-\psi^{(k)}) +p^{(0)}
F^{(k)}
\left(\frac{p^{(1)}}{p^{(0)}},\dots,\frac{p^{(k-1)}}{p^{(0)}}\right), \quad k\ge 2,
\end{split}
\ee
and 
\be
\begin{split}
R^{(0)}&=r_0(n^{(0)})(n^{(0)}p^{(0)}-1), \\
R^{(1)}&=s_n(n^{(0)},p^{(0)})n^{(1)}+s_p(n^{(0)})p^{(1)}-\frac{\tau_n p_T+\tau_p n_T}{\tau_p^2 n^{(0)2}}(n^{(0)}p^{(0)}-1),\\
R^{(k)}&=s_n(n^{(0)},p^{(0)})n^{(k)}+s_p(n^{(0)})p^{(k)} \\
&+F_R^{(k)}(n^{(0)},\dots,n^{(k-1)},p^{(0)},\dots,p^{(k-1)}), \quad k\ge 2,
\end{split}
\ee 
in which 
\[
s_n(n,p)=
C_d p+2 C_{n} np-C_{n}+\frac{1}{\tau_{p}n^2},
\quad
s_p(n)=
C_d n + C_n n^2 + \frac{1}{\tau_{p}},
\]
and $F^{(k)}$, and $F_R^{(k)}$ are some multivariate polynomial functions.  
\end{lemma}
{\bf Proof.}
Let us consider an analytic function $f(z)$, and a function defined by a power series $g(\delta)=\sum_{k\ge0}a_k\delta^k$. We can expand $f(g(\delta))$ in power series of $\delta$,
\[
f(g(\delta))=f(g(0))+\sum_{k\ge1}\left[\frac{\mathrm{d}^{k}}{\mathrm{d}\delta^{k}}f(g(\delta))\right]_{\delta=0}\frac{\delta^{k}}{k!}.
\]
Using Fa\`a di Bruno's formula, we can write the derivatives as
\[
\frac{\mathrm{d}^{k}}{\mathrm{d}\delta^{k}}f(g(\delta))
=\sum_{j_1+2j_2+\cdots k j_k=k} \frac{k!f^{(j)}(g(\delta))}{j_1!j_2!\cdots j_k!}\left(\frac{g'(\delta)}{1!}\right)^{j_1}\left(\frac{g''(\delta)}{2!}\right)^{j_2}\cdots\left(\frac{g^{(k)}(\delta)}{k!}\right)^{j_k},
\]
where  the sum runs over all nonnegative integers $j_1, j_2, \dots, j_k$ which satisfies the Diophantine equation written below the summation, and $j=j_1+j_2+\cdots+j_k$.
Recalling the definition of $g(\delta)$, it is simple to obtain
\[
g^{(i)}(\delta)=\sum_{k\ge i} k(k-1)\cdots(k-i+1) a_{k}\delta^{k-i},
\quad
g^{(i)}(0)=i!\,a_{i},
\]
and thus we have
\be
f(g(\delta))=f(a_0)+\sum_{k\ge1}\sum_{j_1+2j_2+\cdots k j_k=k} \frac{f^{(j)}(a_0)}{j_1!j_2!\cdots j_k!}a_{1}^{j_1}a_{2}^{j_2}\cdots a_{k}^{j_k}\delta^{k}.
\label{faaexp}
\ee
Applying \eqref{faaexp} to $n$, with $f(z)= e^z$, $a_k=\psi^{(k)}-\varphi_n^{(k)}$, and observing that $f^{(k)}(z)=f(z)$ for all $k\ge0$, we get
\be
n = \exp\left(\psi^{(0)}-\varphi_n^{(0)}\right) \left(1 +\sum_{k\ge1}\;\sum_{j_1+\cdots+k j_k=k}\frac{\left(\psi^{(1)}-\varphi_\n^{(1)}\right)^{j_1}\cdots\left(\psi^{(k)}-\varphi_\n^{(k)}\right)^{j_k}}{j_1!\cdots j_k !} \delta^k\right).
\ee
Thus we have
\beaa
n^{(0)} &=& \exp\left(\psi^{(0)}-\varphi_n^{(0)}\right),
\\
n^{(k)} &=& \sum_{j_1+\cdots+k j_k=k} \frac{n^{(0)}}{j_1!\cdots j_k !}\left(\psi^{(1)}-\varphi_\n^{(1)}\right)^{j_1}\cdots\left(\psi^{(k)}-\varphi_\n^{(k)}\right)^{j_k}
\eeaa
In particular, for $k=1$ we have
\[
n^{(1)}=n^{(0)}\left(\psi^{(1)}-\varphi_n^{(1)}\right),
\]
which yields $f^{(1)}=0$, and
\[
\psi^{(1)}-\varphi_n^{(1)}=\frac{n^{(1)}}{n^{(0)}}.
\]
For $k=2$, the condition $j_1+2j_2=2$ is satisfied only by $(j_1,j_2)=(0,1)$, $(2,0)$, and we have
\[
n^{(2)}=n^{(0)}\left(\psi^{(2)}-\varphi_n^{(2)}\right) +\frac{n^{(0)}}{2}\left(\psi^{(1)}-\varphi_\n^{(1)}\right)^{2} =n^{(0)}\left(\psi^{(2)}-\varphi_n^{(2)})\right) +\frac{n^{(0)}}{2}\left(\frac{n^{(1)}}{n^{(0)}}\right)^{2},
\]
that is
\[
n^{(2)}=n^{(0)}\left(\psi^{(2)}-\varphi_n^{(2)}\right) +n^{(0)}F^{(2)}\left(\frac{n^{(1)}}{n^{(0)}}\right),
\]
with $F^{(2)}(r)=\frac12 r^2$. Proceeding by iteration, we assume that
\[
n^{(j)}=n^{(0)}\left(\psi^{(j)}-\varphi_n^{(j)}\right) +n^{(0)}F^{(j)}
\left(\frac{n^{(1)}}{n^{(0)}},\dots,\frac{n^{(j-1)}}{n^{(0)}}\right), \quad 1\le j <k.
\]
We observe that the only solution to $j_1+2 j_2+\cdots+k j_k=k$ with $j_k>0$ is $(j_1,j_2,\dots,j_k)=(0,0,\dots,1)$. Then we can write
\beaa
n^{(k)} &=& n^{(0)}\left(\psi^{(k)}-\varphi_\n^{(k)}\right) \\
&&
+\sum_{j_1+\cdots+(k-1) j_{k-1}=k} \frac{n^{(0)}}{j_1!\cdots j_{k-1} !}\left(\psi^{(1)}-\varphi_\n^{(1)}\right)^{j_1}\cdots\left(\psi^{(k-1)}-\varphi_\n^{(k-1)}\right)^{j_{k-1}},
\eeaa
where the summation in intended over all indices, bearing in mind that $j_k=0$. Hence, $j_k!=1$ and $a_k^{j_k}=1$. 
From the induction hypothesis we have
\[
\psi^{(j)}-\varphi_n^{(j)}=\frac{n^{(j)}}{n^{(0)}} -F^{(j)}\left(\frac{n^{(1)}}{n^{(0)}},\dots,\frac{n^{(j-1)}}{n^{(0)}}\right), \quad 1\le j <k,
\]
which yields
\[
n^{(k)} = n^{(0)}\left(\psi^{(k)}-\varphi_\n^{(k)}\right) +{n^{(0)}}F^{(k)}\left(\frac{n^{(1)}}{n^{(0)}},\dots,\frac{n^{(k-1)}}{n^{(0)}}\right),
\]
with
\[\begin{split}    
F^{(k)}(z_1,\dots,z_{k-1})&=
\\
\sum_{j_1+\cdots+(k-1) j_{k-1}=k}& \frac{1}{j_1!\cdots j_{k-1} !}z_1^{j_1}\cdots\left(z_{k-1}-f_\n^{(k-1)}(z_1,\dots,z_{k-2})\right)^{j_{k-1}}.
\end{split}\]
Similarly, we can prove the expansion of $p$.
For the expansion of $R$, first we consider the expansion of $r_{\delta}(n,p)$. We apply \eqref{faaexp} to the function $f(z)=z^{-1}$, with $a_0=\tau_p n^{(0)}$, $a_1=\tau_p n^{(1)}+\tau_p n_T+\tau_n p_T$, $a_k=\tau_p n^{(k)}+\tau_n p^{(k-2)}$, $k\ge2$, and observe that $f^{(k)}(z)=(-1)^k k! z^{-k-1}$. Then
\[\begin{split}  
\frac{1}{\tau_p(n+\delta n_T)+\tau_n(  p+\delta p_T)}&=
\\ a_0^{-1}&\left(1+\sum_{k\ge1}\,\sum_{j_1+2j_2+\cdots k j_k=k} \frac{(-1)^jj!}{j_1!j_2!\cdots j_k!} \frac{a_{1}^{j_1}a_{2}^{j_2}\cdots a_{k}^{j_k}}{a_0^{j}}\delta^{k}\right),
\end{split}\]
and we find
\[
r_\delta\left(\sum_{k}n^{(k)}\delta^k,\delta^2\sum_{k}p^{(k)}\delta^k\right) = \sum_{k\ge0}r^{(k)}\delta^k,
\]
with
\beaa
r^{(0)}=r_0(n^{(0)}), 
\quad
r^{(k)}=c(n^{(0)}) n^{(k)}+F_{r}^{(k)}, \quad k\ge 0,
\eeaa
where
\beaa
&&
r_{0}(n) := C_d+C_n n+\frac{1}{\tau_p n},
\quad
c(n) := C_n-\frac{1}{\tau_p n^{2}},
\quad
F_{r}^{(1)} := -\frac{\tau_p n_T+\tau_n p_T}{\tau_p^2n^{(0)2}},
\\ &&
F_r^{(k)}:= \left(C_p-\frac{\tau_{n}}{\tau_{p}n^{(0)2}}\right)p^{(k-2)}
\\ &&\qquad\qquad 
+\sum_{{j_1+2j_2+\cdots + (k-1) j_{k-1}=k}} \frac{(-1)^jj!}{j_1!j_2!\cdots j_{k-1}!} \frac{a_{1}^{j_1}a_{2}^{j_2}\cdots a_{k-1}^{j_{k-1}}}{a_0^{j}}, \quad k\ge2.
\eeaa
Using the Cauchy product
\[
\frac{np}{\delta^2}=\sum_{k\ge0} \sum_{j=0}^{k} n^{(j)}p^{(k-j)} \delta^k,
\]
and combining it with the expansion of $r_\delta$, we get
\[
R=\delta^2 \sum_{k\ge0}R^{(k)}\delta^k,
\quad
R^{(k)}:=\sum_{j=0}^{k}r^{(j)}\sum_{i=0}^{k-j}n^{(i)}p^{(k-j-i)} -r^{(k)},
\]
which, thanks to the identities $s_n(n,p)=r_{0}(n)p+c(n)(np-1)$, $s_p(n)=r_{0}(n)n$,  gives
\[
R^{(0)}=r_{0}(n^{(0)})(n^{(0)}p^{(0)}-1),
\quad
R^{(k)}=s_n(n^{(0)},p^{(0)})n^{(k)}+s_p(n^{(0)})p^{(k)}+
F_R^{(k)},
\quad k\ge1,
\]
where
\beaa
F_{R}^{(1)} &:=& F_{r}^{(1)}(n^{(0)}p^{(0)}-1),
\\
F_R^{(k)} &:=& F_{r}^{(k)}(n^{(0)}p^{(0)}-1) +r^{(0)}\sum_{i=1}^{k-1}n^{(i)}p^{(k-i)}+\sum_{j=1}^{k-1}r^{(j)}\sum_{i=0}^{k-j}n^{(i)}p^{(k-j-i)}.
\eeaa
This proves the final part of the lemma.
\qed

Using the above expansions in \eqref{eq:vR-scaled-model_fin} we find:
\be
\begin{cases}
-\lambda^2 \nabla \cdot (\varepsilon \nabla \psi^{(k)}) = 
F_{\psi}^{(k)},
\\[1ex]
- \nabla\cdot ( \mu_\n  n^{(0)} \nabla \varphi_\n^{(k)} +
 \mu_\n  n^{(k)} \nabla \varphi_\n^{(0)}) = F_{\varphi_n}^{(k)}, \qquad
\\[1ex]
- \nabla\cdot (\mu_\p  p^{(0)} \nabla \varphi_{\p}^{(k)} + \mu_\p p^{(k)}  \nabla \varphi_{\p}^{(0)}) 
=  F_{\varphi_p}^{(k)}, \qquad
\end{cases}
\label{eq:system-after-expantion}
\ee
where 
\[
F_{\psi}^{(0)} = -n^{(0)} +C(\boldsymbol{x}), \qquad 
F_{\psi}^{(1)} = -n^{(1)}, \qquad 
F_{\psi}^{(k)} = p^{(k-2)}-n^{(k)}, \qquad k >  1, 
\]
\[
F_{\varphi_n}^{(0)} =  F_{\varphi_n}^{(1)} = 0, \qquad 
F_{\varphi_n}^{(2)} = R^{(0)} - G(\boldsymbol{x}) + \nabla \cdot (
\mu_\n  n^{(1)} \nabla \varphi_\n^{(1)}), \qquad 
\] 
\[
F_{\varphi_n}^{(k)} = R^{(k)} + 
\sum_{j=1}^{k-1}\nabla\cdot (\mu_\n  n^{(j)} \nabla \varphi_{\n}^{(k-j)}), \qquad k > 2, 
\]
\[
F_{\varphi_p}^{(0)} = 2(G(\boldsymbol{x})- R^{(0)}), \qquad 
F_{\varphi_p}^{(k)} = -R^{(k)} + 
\sum_{j=1}^{k-1}\nabla\cdot (\mu_\p  p^{(j)} \nabla \varphi_{\p}^{(k-j)}), \qquad k > 0,
\]
and applying the expansion to the coupling condition \eqref{eq:coupling} we find:
\bea
{u}_D^{(k)}   =
\begin{cases}
     \displaystyle
     -{\mathcal{R}}\displaystyle\int_{{\Gamma}_{D_2}}  \mu_n n^{(0)}\frac{\partial \varphi_n^{(0)}}{\partial \boldsymbol{\nu}} d{\sigma}(\boldsymbol{x}), \quad  k=0,
     \\[2ex]
     \displaystyle-{\mathcal{R}}\displaystyle\int_{{\Gamma}_{D_2}}  \left(\mu_n n^{(0)}\frac{\partial \varphi_n^{(1)}}{\partial \boldsymbol{\nu}} + \mu_n n^{(1)}\frac{\partial \varphi_n^{(0)}}{\partial \boldsymbol{\nu}} \right) d{\sigma}(\boldsymbol{x}), \quad k=1,
       \\[2ex]
      \displaystyle
      -{\mathcal{R}}\displaystyle\int_{{\Gamma}_{D_2}}  \left(\mu_n \sum_{j=0}^{k}n^{(j)}\frac{\partial \varphi_n^{(k-j)}}{\partial \boldsymbol{\nu}} +  \mu_p \sum_{j=0}^{k-2}p^{(j)}\frac{\partial \varphi_p^{(k-2-j)}}{\partial \boldsymbol{\nu}} \right) d{\sigma}(\boldsymbol{x}), \quad k>1.
\end{cases}
\eea

We have
\begin{equation}
  \frac{\partial {\psi}^{(k)}}{\partial \boldsymbol{\nu}}=
  \frac{\partial {\varphi}_\n^{(k)}}{\partial \boldsymbol{\nu}}=
  \frac{\partial {\varphi}_\p^{(k)}}{\partial \boldsymbol{\nu}}=0,  \quad \text{on }\Gamma_N, \quad k\ge0,
  \label{eq:neumann-data-k}
\end{equation}
\begin{equation}
  \begin{split}
  &{\psi}^{(0)} - {\psi}_0= {\varphi}_\n^{(0)} -{\varphi}_0 = {\varphi}_\p^{(0)} -{\varphi}_0 = 0,\quad \text{on }\quad\Gamma_{D_1}, 
  \\& 
  {\psi}^{(0)} - {\psi}_0= {\varphi}_\n^{(0)} - {\varphi}_0= {\varphi}_\p^{(0)} -{\varphi}_0= {u}_{D}^{(0)},  \quad \text{on }\quad\Gamma_{D_2},
  \label{eq:Dirichlet-data-0}
  \end{split}
\end{equation}
\begin{equation}
  \begin{split}
  &{\psi}^{(k)} = {\varphi}_\n^{(k)} = {\varphi}_\p^{(k)} = 0,\quad \text{on }\quad\Gamma_{D_1}, 
  \\& 
  {\psi}^{(k)} = {\varphi}_\n^{(k)} = {\varphi}_\p^{(k)} = {u}_{D}^{(k)},  \quad \text{on }\quad\Gamma_{D_2}, \quad k\ge1.
  \label{eq:Dirichlet-data-k}
  \end{split}
\end{equation}

We focus our attention on the leading order terms for the expansion of the carrier concentrations $n, p$ and of the potential $u_D$. In the next section, we will show that 
\be
 \varphi_n^{(0)} - \varphi_0 = 
 \psi^{(1)} =\varphi_n^{(1)} = 0, \qquad u_D^{(0)} =u_D^{(1)} = 0.
\label{eq:vanishing-unknowns}
\ee
Thus, we concentrate on the corresponding leading order terms for the potentials 
\[
\psi = \psi^{(0)} + \delta^2 \psi^{(2)} + O(\delta^3), \quad  
\varphi_n  = \varphi_0 + \delta^2 \varphi_n^{(2)} + O(\delta^3),
\quad  
\varphi_p  = \varphi_p^{(0)}  + O(\delta),
\]
which lead to 
\[\begin{split}   
n &= e^{\psi^{(0)}-  \varphi_0 }\left( 1+ \delta^2 (\psi^{(2)} - \varphi_n^{(2)}  )\right)+ O(\delta^3),
\quad 
p = \delta^2 e^{ \varphi_p^{(0)}  - \psi^{(0)}}+ O(\delta^3),
\\
u_D &= \delta^2 u_D^{(2)} + O(\delta^3). 
\end{split}
\]
In the following section, we will show that the resulting simplified model depends only on the five unknowns $\psi^{(0)}$, $\varphi_p^{(0)}$,  $\psi^{(2)}$, $\varphi_n^{(2)}$ and $u_D^{(2)}$ which satisfy the following equations
\be
\begin{cases}
-\lambda^2 \nabla \cdot (\varepsilon \nabla \psi^{(0)}) &= 
-e^{\psi^{(0)}-  \varphi_0 }+C(\boldsymbol{x}),
\\[1ex]
- \nabla\cdot (\mu_\p  e^{ \varphi_p^{(0)}  - \psi^{(0)}} \nabla \varphi_{\p}^{(0)}) &=   
-r_{0}(e^{\psi^{(0)}-  \varphi_0 })( e^{ \varphi_p^{(0)}  - \varphi_0 } -1) + G(\boldsymbol{x}), 
\\[1ex]
-\lambda^2 \nabla \cdot (\varepsilon \nabla \psi^{(2)}) & = 
e^{ \varphi_p^{(0)}  - \psi^{(0)}}- e^{\psi^{(0)}-  \varphi_0} ({\psi^{(2)}-  \varphi_n^{(2)} }),
\\[1ex]
- \nabla\cdot (\mu_\n e^{   \psi^{(0)} - \varphi_0}  \nabla \varphi_{\n}^{(2)}) &=   r_{0}(e^{\psi^{(0)}-  \varphi_0 })( e^{ \varphi_p^{(0)}  - \varphi_0 } -1) - G(\boldsymbol{x}), 
\end{cases}
\label{eq:final-system}
\ee
with coupling condition 
\be
u_D^{(2)} =   -{\mathcal{R}}\displaystyle\int_{{\Gamma}_{D_2}}  \left(\mu_n  e^{\psi^{(0)}-  \varphi_0 } 
\frac{\partial \varphi_n^{(2)}}{\partial \boldsymbol{\nu}} +  \mu_p   e^{ \varphi_p^{(0)} - \psi^{(0)}} \frac{\partial \varphi_p^{(0)}}{\partial \boldsymbol{\nu}} \right) d{\sigma}(\boldsymbol{x}),
\label{eq:final-coupling-condition}
\ee
and boundary conditions 
\be
  \frac{\partial \psi^{(0)} }{\partial \boldsymbol{\nu}}=
  \frac{\partial \varphi_p^{(0)} }{\partial \boldsymbol{\nu}}=
  \frac{\partial \psi^{(2)} }{\partial \boldsymbol{\nu}}=
  \frac{\partial \varphi_n^{(2)}}{\partial \boldsymbol{\nu}}=
  0,  \quad \text{on }\Gamma_N,
  \label{eq:final-neumann-data}
  \ee
\be
\begin{split}
&{\psi}^{(0)} - {\psi}_0=  {\varphi}_\p^{(0)} -{\varphi}_0 = \psi^{(2)}  =\varphi_n^{(2)}=  0,\quad \text{on }\quad\Gamma_{D_1}, 
  \\&
  {\psi}^{(0)} - {\psi}_0= {\varphi}_\p^{(0)} - {\varphi}_0=0, \quad \psi^{(2)} =  {\varphi}_\n^{(2)}= {u}_{D}^{(2)},  \quad \text{on }\quad\Gamma_{D_2}.
  \label{eq:final-Dirichlet-data-uD}    
\end{split}
\ee
We observe that \eqref{eq:final-system}$_1$ and the corresponding boundary conditions result into a nonlinear elliptic bounday value problem for $\psi^{(0)}$. Once $\psi^{(0)}$ is known, equation \eqref{eq:final-system}$_2$ and the corresponding boundary conditions become a nonlinear elliptic boundary value problem for $\varphi_p^{(0)}$. Then, equations \eqref{eq:final-system}$_4$ and \eqref{eq:final-coupling-condition} and the corresponding boundary conditions are a coupled problem for $\varphi_n^{(2)}$ and $u_D^{(2)}$.
Finally, equation \eqref{eq:final-system}$_3$ and the corresponding boundary conditions constitute a linear boundary value problem for $\psi^{(2)}$.

In the following section, we prove the validity of \eqref{eq:vanishing-unknowns} and the well-posedness of \eqref{eq:final-system}--\eqref{eq:final-Dirichlet-data-uD}.

\if 
at order 1:
\[
\begin{split}
-\lambda^2 \nabla \cdot (\varepsilon \nabla \psi^{(1)}) &= 
-n^{(1)},
\\
- \nabla\cdot ( \mu_\n  n^{(0)} \nabla \varphi_\n^{(1)}+ \mu_\n  n^{(1)} \nabla \varphi_\n^{(0)}) &= 0, \qquad
\\
- \nabla\cdot (\mu_\p  p^{(0)} \nabla \varphi_{\p}^{(1)}+\mu_\p  p^{(1)} \nabla \varphi_{\p}^{(0)}) &= -s_n(n^{(0)},p^{(0)})n^{(1)}-s_p(n^{(0)})p^{(1)}\\ &+\frac{\tau_n p_T+\tau_p n_T}{\tau_p^2 n^{(0)2}}(n^{(0)}p^{(0)}-1), \qquad
\end{split}
\]
at order 2:
\[
\begin{split}
-\lambda^2 \nabla \cdot (\varepsilon \nabla \psi^{(2)}) &= 
p^{(0)}-n^{(2)},
\\
- \nabla\cdot ( \mu_\n  n^{(0)} \nabla \varphi_\n^{(2)}+\mu_\n  n^{(2)} \nabla \varphi_\n^{(0)} ) &= \nabla\cdot (\mu_\n  n^{(1)} \nabla \varphi_{\n}^{(1)})+R^{(0)} -  G(\boldsymbol{x}), \qquad
\\
- \nabla\cdot (\mu_\p  p^{(0)} \nabla \varphi_{\p}^{(2)}+\mu_\p  p^{(2)} \nabla \varphi_{\p}^{(0)}) &= -s_n(n^{(0)},p^{(0)})n^{(2)}-s_p(n^{(0)})p^{(2)}\\
&+\nabla\cdot (\mu_\p  p^{(1)} \nabla \varphi_{\p}^{(1)})-F_{R}^{(2)}, \qquad
 \end{split}
\]
at order $k>2$:
\[
\begin{split}
-\lambda^2 \nabla \cdot (\varepsilon \nabla \psi^{(k)}) &= 
p^{(k-2)}-n^{(k)},
\\
- \nabla\cdot ( \mu_\n  n^{(0)} \nabla \varphi_\n^{(k)}+\mu_\n  n^{(k)} \nabla \varphi_\n^{(0)} ) &= \sum_{j=1}^{k-1}\nabla\cdot (\mu_\n  n^{(j)} \nabla \varphi_{\n}^{(k-j)})+R^{(k-2)}, \qquad
\\
- \nabla\cdot (\mu_\p  p^{(0)} \nabla \varphi_{\p}^{(k)}+\mu_\p  p^{(k)} \nabla \varphi_{\p}^{(0)}) &= -s_n(n^{(0)},p^{(0)})n^{(k)}-s_p(n^{(0)})p^{(k)}\\
&+\sum_{j=1}^{k-1}\nabla\cdot (\mu_\p  p^{(j)} \nabla \varphi_{\p}^{(k-j)})-F_{R}^{(k)}. \qquad
 \end{split}
\]
\fi

\section{Existence and uniqueness of solutions}
\label{sec:existence} 
Based on the previous discussion we need to keep only the zeroth order term for $\varphi_p$ in order to have second order expansion for the carrier densities. Hence, we focus our attention on   
\beaa
\psi &=& \psi^{(0)} + \delta \psi^{(1)} + \delta^2 \psi^{(2)} + O(\delta^3),  
\\
\varphi_n  & = & \varphi_n^{(0)}  + \delta \varphi_n^{(1)}  + \delta^2 \varphi_n^{(2)} + O(\delta^3),
\\
\varphi_p & =&  \varphi_p^{(0)}   + O(\delta),
\\
u_D  & = & u_D^{(0)}  + \delta u_D^{(1)}  + \delta^2 u_D^{(2)} + O(\delta^3),
\eeaa
and consider the resulting equations for each order separately in order to prove  the validity of \eqref{eq:vanishing-unknowns} and show existence and uniqueness for the system \eqref{eq:final-system}--\eqref{eq:final-Dirichlet-data-uD}.

\subsection{Existence and uniqueness of solutions for the zero order problem} 
At order 0, the nonlinear equations for $\psi^{(0)}$, $\varphi_\n^{(0)}$, $\varphi_\p^{(0)}$ and $u_D^{(0)}$, spelled out in full detail, are:
\be
\begin{cases}
-\lambda^2 \nabla \cdot \left(\varepsilon \nabla \psi^{(0)}\right) &= 
- e^{\psi^{(0)}-\varphi_n^{(0)}}+C(\boldsymbol{x}),
\label{eq:system-order0}
\\
- \nabla\cdot \left( \mu_\n    e^{\psi^{(0)}-\varphi_n^{(0)}} \nabla \varphi_\n^{(0)} \right) &= 0, 
\\
- \nabla\cdot \left(\mu_\p    e^{\varphi_p^{(0)}-\psi^{(0)}} \nabla \varphi_{\p}^{(0)}\right) &=   G(\boldsymbol{x})-r_{0}\left(  e^{\psi^{(0)}-\varphi_n^{(0)}}\right)\left(e^{\varphi_p^{(0)}-\varphi_n^{(0)}}-1\right), 
\end{cases}
\ee
\be
  \frac{\partial {\psi}^{(0)}}{\partial \boldsymbol{\nu}}=
  \frac{\partial {\varphi}_\n^{(0)}}{\partial \boldsymbol{\nu}}=
  \frac{\partial {\varphi}_\p^{(0)}}{\partial \boldsymbol{\nu}}=0,  \quad \text{on }\Gamma_N,
\label{eq:order0-neumann}
\ee
\be
\begin{split}
{\psi}^{(0)} - {\psi}_0= {\varphi}_\n^{(0)}  -{\varphi}_0= {\varphi}_\p^{(0)} -{\varphi}_0 &= 0, \quad \text{on }\quad\Gamma_{D_1}, 
\label{eq:order0-dirichlet1-2}
\\
{\psi}^{(0)} - {\psi}_0= {\varphi}_\n^{(0)} -{\varphi}_0 = {\varphi}_\p^{(0)}  -{\varphi}_0&= {u}_{D}^{(0)}, \quad \text{on }\quad\Gamma_{D_2},
\end{split}
\ee
\be
{u}_D^{(0)}   =  
      -{\mathcal{R}}\displaystyle\int_{{\Gamma}_{D_2}}  \mu_n  e^{\psi^{(0)}-\varphi_n^{(0)}}\frac{\partial \varphi_n^{(0)}}{\partial \boldsymbol{\nu}} d{\sigma}(\boldsymbol{x}).
\label{eq:order0-uD}
\ee

\begin{theorem}
Assuming that $C$, $G\in L^\infty(\Omega)$, the zero-order equations \eqref{eq:system-order0} with boundary conditions \eqref{eq:order0-neumann}--\eqref{eq:order0-dirichlet1-2} and coupling condition \eqref{eq:order0-uD} admit a unique solution in $H^1(\Omega) \cap L^\infty(\Omega)$ for $\psi^{(0)}$, $\varphi_n^{(0)}$, $\varphi_p^{(0)}$, with $u_D^{(0)}\in\mathbb{R}$. Moreover this solution satisfies the identities
\be
u_D^{(0)}=0, \quad \varphi_n^{(0)} ={\varphi}_0,
\label{eq:order0-uD-phin}
\ee
and the estimates
\bea
\min\left\{\inf_{\Gamma_D}\psi_0,\, \varphi_0 +\ln {\underline{C}}\right\} =: \underline{\psi}^{(0)}  &\le \psi^{(0)} \le & \overline{\psi}{}^{(0)} := \max\left\{\sup_{\Gamma_D}\psi_0,\,  \varphi_0 + \ln  {\overline{C}} \right\},\qquad
\label{eq:bounds-psi0}
\\
\varphi_0+ \underline{\varphi}_p &\le \varphi_p^{(0)}  \le & \varphi_0+\overline{\varphi}_p
\label{eq:bounds-phip0}
\eea
where  
\be 
\begin{split}  
\underline{\varphi}_p &:=  \ln (\underline{r}/\overline{r}), \quad \overline{\varphi}_p:=\ln \frac{\overline{r}+\overline{G}}{\underline{r}}, \quad 
\displaystyle\underline{C}:=\inf_\Omega C(x), \quad
\\
\displaystyle\overline{C}&:=\sup_\Omega C(x) \equiv 1, \quad 
\displaystyle\overline{G}:=\sup_\Omega G(x),
\label{eq:phip-G-C}
\end{split}
\ee
in which 
\be
\underline{r} = C_d+C_{\n}e^{\underline{\psi}^{(0)}-\varphi_0}+\frac1{\tau_p}e^{\varphi_0-\overline{\psi}^{(0)}},
\qquad
\overline{r} = C_d+C_{\n}e^{\overline{\psi}^{(0)}-\varphi_0}+\frac1{\tau_p}e^{\varphi_0-\underline{\psi}^{(0)}}.
\label{eq:r-overbar-underbar}
\ee
\end{theorem}
{\bf Proof.}
First, we prove that any solution of \eqref{eq:system-order0}--\eqref{eq:order0-uD} satisfies \eqref{eq:order0-uD-phin}. We multiply \eqref{eq:system-order0}$_2$ by $\varphi_n^{(0)}-\varphi_0$, integrate over $\Omega$ and use integration by parts to obtain
\beaa
0 &=& -\int_{\Omega}(\varphi_n^{(0)}-\varphi_0)  \nabla\cdot \left( \mu_\n   e^{\psi^{(0)}-\varphi_n^{(0)}} \nabla \varphi_\n^{(0)} \right)\,\mathrm{d}\boldsymbol{x}
\\ 
&=& -\int_{\partial\Omega}(\varphi_n^{(0)}-\varphi_0)\mu_\n   e^{\psi^{(0)}-\varphi_n^{(0)}} \frac{\partial \varphi_\n^{(0)}}{\partial \boldsymbol{\nu}} \,\mathrm{d}\sigma(\boldsymbol{x}) +\int_{\Omega} \mu_\n   e^{\psi^{(0)}-\varphi_n^{(0)}} |\nabla \varphi_\n^{(0)}|^2 \,\mathrm{d}\boldsymbol{x}
\\ 
&=& \frac{1}{\mathcal{R}}|u_D^{(0)}|^2 +\int_{\Omega} \mu_\n   e^{\psi^{(0)}-\varphi_n^{(0)}} |\nabla \varphi_\n^{(0)}|^2 \,\mathrm{d}\boldsymbol{x},
\eeaa
which implies \eqref{eq:order0-uD-phin}.
Then $\psi^{(0)}$ must solve
\be
\begin{cases}
-\lambda^2 \nabla \cdot \left(\varepsilon \nabla \psi^{(0)}\right) &= 
- e^{\psi^{(0)}-\varphi_0}+C(\boldsymbol{x}),
\label{eq:order0-poisson-bis}
\\
  \frac{\partial {\psi}^{(0)}}{\partial \boldsymbol{\nu}}&=0,  \quad \text{on }\Gamma_N,
\\ 
{\psi}^{(0)} - {\psi}_0 &= 0, \quad \text{on }\quad\Gamma_{D_1}\cup\Gamma_{D_2},
\end{cases}
\ee
and, with standard methods it is simple to prove that a solution exists uniquely and satisfies the bounds \eqref{eq:bounds-psi0}.

Finally, $\varphi_p^{(0)}$ must solve
\be
\begin{cases}
    - \nabla\cdot \left(\mu_\p   e^{\varphi_p^{(0)}-\psi^{(0)}} \nabla \varphi_{\p}^{(0)}\right) &=  G(\boldsymbol{x})-r_{0}\left( e^{\psi^{(0)}-\varphi_0}\right)\left(e^{\varphi_p^{(0)}-\varphi_0}-1\right), 
\label{eq:order0-holes-bis}
\\
  \frac{\partial {\varphi}_\p^{(0)}}{\partial \boldsymbol{\nu}}&=0,  \quad \text{on }\Gamma_N,
\\ 
{\varphi}_\p^{(0)} -\varphi_0&= 0, \quad \text{on }\quad\Gamma_{D_1}\cup\Gamma_{D_2}.
\end{cases}
\ee
Equation \eqref{eq:order0-holes-bis} can be written in the form
\be
- \nabla\cdot \left(\mu_\p  e^{\varphi_p^{(0)}-\psi^{(0)}} \nabla \varphi_{\p}^{(0)}\right)
=f(\varphi_p^{(0)},\boldsymbol{x}),
\label{eq:order0-dirichlet2-bisss}
\ee
with
\[
f(\varphi,\boldsymbol{x}) :=   G(\boldsymbol{x})+ r_0\left(  e^{ \psi^{(0)}(\boldsymbol{x}) -\varphi_0}\right) -r_0\left(  e^{ \psi^{(0)}(\boldsymbol{x})-\varphi_0}\right)e^{\varphi-\varphi_0}.
\]
It is simple to verify that
\[
\underline{r}\le r_0\left( e^{ \psi^{(0)}(\boldsymbol{x})-\varphi_0}\right) \le \overline{r},
\]
thus, we have
\be
-\overline{r}\exp{\varphi}+\underline{r}=:\underline{f}(\varphi) \le f(\varphi,x) 
\le \overline{f}(\varphi):=-\underline{r}\exp{\varphi}+\overline{r}+ \overline{G}.
\ee
The function $\underline{f}(\varphi)$ vanishes if and only if $\varphi-\varphi_0=\underline{\varphi}_{p}:=\ln (\underline{r}/\overline{r})$. Similarly, the function $\overline{f}(\overline{\varphi})$ vanishes if and only if $\varphi-\varphi_0=\overline{\varphi}_{p}:=\ln \frac{\overline{r}+\overline{G}}{\underline{r}}$.
It follows that $\varphi_{p}^{(0)}=\varphi_0+\underline{\varphi}_{p}$ is a lower solution and $\varphi_{p}^{(0)}=\varphi_0+\overline{\varphi}_{p}$ is an upper solution for \eqref{eq:order0-dirichlet2-bisss}.
Then, it follows that $\varphi_\p^{(0)}$ satisfies the bounds \eqref{eq:bounds-phip0}. \qed

\subsection{Existence and uniqueness of solutions for the first order problem} 
Next we consider equations at order 1, taking into account the results of Theorem 1. The resulting linear equations for $\psi^{(1)}$, $\varphi_n^{(1)}$, $u_D^{(1)}$ are:
\be
\begin{cases}  
&
-\lambda^2 \nabla \cdot (\varepsilon \nabla \psi^{(1)}) =
-n^{(0)}(\psi^{(1)}-\varphi_n^{(1)}),
\label{eq:system-order1}
\\&
- \nabla\cdot ( \mu_\n  n^{(0)} \nabla \varphi_\n^{(1)}) =
0, 
\\&
  \frac{\partial {\psi}^{(1)}}{\partial \boldsymbol{\nu}}=
  \frac{\partial {\varphi}_\n^{(1)}}{\partial \boldsymbol{\nu}}=0,  \quad \text{on }\Gamma_N,
\\ &
 {\psi}^{(1)} = {\varphi}_\n^{(1)} = 0,  \quad \text{on }\quad\Gamma_{D_1}, 
\qquad 
 {\psi}^{(1)} = {\varphi}_\n^{(1)}   = {u}_{D}^{(1)}, \quad \text{on }\quad\Gamma_{D_2},
\\& 
{u}_D^{(1)}    =
      -{\mathcal{R}}\displaystyle\int_{{\Gamma}_{D_2}}  \mu_n  e^{\psi^{(0)}}\frac{\partial \varphi_n^{(1)}}{\partial \boldsymbol{\nu}} \,d{\sigma}(\boldsymbol{x}),
\end{cases}
\ee 
with $n^{(0)}$, $p^{(0)}$ given by Lemma 1, that is,
\be
n^{(0)}= e^{\psi^{(0)}-\varphi_0},
\quad 
p^{(0)}= e^{\varphi_{p}^{(0)}-\psi^{(0)}}.
\label{eq:n0-p0}
\ee 
\begin{theorem}
Assuming that $\psi^{(0)}$, $\varphi_\p^{(0)}\in H^1(\Omega) \cap L^\infty(\Omega)$, $\varphi_\n^{(0)}=\varphi_0$, $u_D^{(0)}=0$ is the solution to the zero-order equations \eqref{eq:system-order0}--\eqref{eq:order0-uD} given by Theorem 1, the first order system of equations \eqref{eq:system-order1} admits a unique solution in $H^1(\Omega) \cap L^\infty(\Omega)$ for $\psi^{(1)}$, $\varphi_n^{(1)}$, with $u_D^{(1)}\in\mathbb{R}$. Moreover, this solution is trivial and satisfies the identities
\be
u_D^{(1)}=0, \quad \varphi_n^{(1)}=0, \quad \psi^{(1)}=0.
\label{eq:order1-uD-phin}
\ee

\end{theorem}
{\bf Proof.}
Using the second and the last equations in \eqref{eq:system-order1}$_2$, together with the boundary conditions, with a similar argument to the one used in Theorem 1, it is possible to show that $u_D^{(1)}=0$ and $\varphi_\n^{(1)}=0$. Then the equation for the potential $\psi^{(1)}$ reduces to
\[\begin{cases}
-\lambda^2 \nabla \cdot (\varepsilon \nabla \psi^{(1)}) =
-n^{(0)}\psi^{(1)},
\\
 \frac{\partial {\psi}^{(1)}}{\partial \boldsymbol{\nu}}=0,  \quad \text{on }\Gamma_N,
\quad 
{\psi}^{(1)} = 0, \quad \text{on }\quad\Gamma_{D_1}\cup \Gamma_{D_2},
\end{cases}
\]
which is a linear elliptic homogeneous problem with unique solution $\psi^{(1)}=0$. Thus, \eqref{eq:order1-uD-phin}
is proven. 
\qed 

\subsection{Existence and uniqueness of solutions for the second order problem} 
Finally, taking into account the results of Theorem 1 and Theorem 2, the equations at order 2 decouple in two systems for $\varphi_\n^{(2)} $ and $\psi^{(2)}$, namely 
\be
\begin{cases}
- \nabla\cdot ( \mu_\n n^{(0)} \nabla \varphi_\n^{(2)}  )= r_0(n^{(0)})(n^{(0)}p^{(0)}-1) -  G(\boldsymbol{x})
\quad \text{ in } \Omega,
\\\frac{
  \partial\varphi_\n^{(2)}}{\partial\boldsymbol{\nu}}=0 \quad \text{ on } \Gamma_{N}, 
\\
\varphi_\n^{(2)} =0 \quad \text{ on } \Gamma_{D,1}, 
\qquad
\varphi_\n^{(2)} =u_D^{(2)} \quad \text{ on } \Gamma_{D,2},
\\
{u}_D^{(2)}   =  -{\mathcal{R}}\displaystyle\displaystyle\int_{{\Gamma}_{D_2}}  \left(\mu_n n^{(0)}\frac{\partial \varphi_n^{(2)}}{\partial \boldsymbol{\nu}} +  \mu_p p^{(0)}\frac{\partial \varphi_p^{(0)}}{\partial \boldsymbol{\nu}} \right) d{\sigma}(\boldsymbol{x}),
\end{cases}
\label{eq:varphin2-uD2}
\ee 
and 
\be
\begin{cases}
-\lambda^2 \nabla \cdot (\varepsilon \nabla \psi^{(2)}) = 
p^{(0)}-n^{(0)}(\psi^{(2)}-\varphi_n^{(2)}),
\\
\frac{\partial {\psi}^{(2)}}{\partial \boldsymbol{\nu}}=
0,  \quad \text{on }\Gamma_N,
\\
{\psi}^{(2)}    = 0,\quad \text{on }\quad\Gamma_{D_1}, \qquad 
{\psi}^{(2)} =  {u}_{D}^{(2)},  \quad \text{on }\quad\Gamma_{D_2},
  \end{cases}
  \label{eq:psi2alone}
\ee
where $n^{(0)}$ and $p^{(0)}$ are given by \eqref{eq:n0-p0}.
Once $\varphi_p^{(0)}$ is known, equation \eqref{eq:varphin2-uD2} is a linear elliptic equation for $\varphi_n^{(2)}$, coupled with the condition \eqref{eq:varphin2-uD2}$_4$ for $u_D^{(2)}$. Then, equation \eqref{eq:psi2alone} becomes a linear equation for $\psi^{(2)}$.

\begin{theorem}
Assuming that $\varphi_\n^{(0)}=\varphi_0 $, $\psi^{(0)}$ and $\varphi_\p^{(0)}\in H^1(\Omega) \cap L^\infty(\Omega)$, $u_D^{(0)}=0$ is the solution to the zero-order equations \eqref{eq:system-order0}--\eqref{eq:order0-uD} given by Theorem 1, the second-order system of equations \eqref{eq:varphin2-uD2} admits a unique solution 
$\varphi_n^{(2)} \in H^1(\Omega) \cap L^\infty(\Omega)$,  $u_D^{(2)} \in \mathbb{R} $, and the boundary value problem \eqref{eq:psi2alone} admits a unique solution ${\psi}^{(2)}  \in H^1(\Omega) \cap L^\infty(\Omega)$. 
Moreover this solution satisfies the estimates
\be
\min\{ 0,  \underline{r} -\overline{G}\} -\overline{{u}}_D 
\le 
\varphi_n^{(2)} \le 
\overline{r} +\overline{{u}}_D,
\label{eq:final-estimate-phin2}
\ee 
\be 
\underline{\psi}^{(2)}  
\le 
{\psi}^{(2)}  
\le  
\overline{\psi}{}^{(2)},  
\label{eq:final-estimate-psi2}
\ee
and 
\be
|{u}_D^{(2)}|\leq  \overline{{u}}_D,
\ee
with 
\be\begin{split}
 \overline{{u}}_D &:= {\mathcal{R}} \left(
\| \mu_n n^{(0)}\|_{L^\infty(\Omega)} 
\| \varphi_n^{*} \|_{H^1(\Omega)}
+\| \mu_p p^{(0)}\|_{L^\infty(\Omega)} 
\| \varphi_p^{(0)} \|_{H^1(\Omega)}
\right) \| w\|_{H^1(\Omega)},
\\
 \underline{\psi}^{(2)}  &:= \min\left\{ -\overline{{u}}_D,\inf\frac{p^{(0)}}{n^{(0)}} \right\} + \min\{0, \underline{r} -\overline{G} \}  -\overline{{u}}_D, 
   \\ 
\overline{\psi}{}^{(2)}  &
:= \max\left\{ \overline{{u}}_D,\sup\frac{p^{(0)}}{n^{(0)}}  \right\}+\overline{r} +\overline{{u}}_D,
\label{eq:estimate-uDbar}
\end{split}
\ee 
where $w$ is the solution of the elliptic boundary value problem:  
\be
\begin{cases}
- \nabla\cdot ( \mu_\n n^{(0)} \nabla w  )=0,
\quad \text{ in } \Omega,
\\\frac{
  \partial w}{\partial\boldsymbol{\nu}}=0 \quad \text{ on } \Gamma_{N}, 
\\
w = w_D = \begin{cases}
    0 \quad \text{ on } \Gamma_{D,1}
    \\
    1 \quad \text{ on } \Gamma_{D,2}
\end{cases}
\end{cases}
\label{eq:aux-w}
\ee 
    In \eqref{eq:final-estimate-phin2} and \eqref{eq:final-estimate-psi2}, $\underbar{r},\overline{r}$ and $\overline{G}$ are defined in \eqref{eq:phip-G-C} and \eqref{eq:r-overbar-underbar}.
\end{theorem}
\begin{proof}
By standard results, see for example  \cite{protter}, there exists a unique solution $w \in H^1(\Omega) \cap L^\infty(\Omega)$ of problem \eqref{eq:aux-w} such that 
\[
0\le w \le 1.
\]
This allows to rewrite the last equation in \eqref{eq:varphin2-uD2}. Assuming that $w$ satisfies problem \eqref{eq:aux-w} and using integration py parts, we can write
\begin{align*}
    \int_{\Gamma_{D,2}} &\!\!\!\!\! 
    w_D \left(\mu_n n^{(0)}\frac{\partial \varphi_n^{(2)}}{\partial \boldsymbol{\nu}} +  \mu_p p^{(0)}\frac{\partial \varphi_p^{(0)}}{\partial \boldsymbol{\nu}} \right) d{\sigma}(\boldsymbol{x}) \\&= \int_{\partial\Omega} w \left(\mu_n n^{(0)}\frac{\partial \varphi_n^{(2)}}{\partial \boldsymbol{\nu}} +  \mu_p p^{(0)}\frac{\partial \varphi_p^{(0)}}{\partial \boldsymbol{\nu}} \right) d{\sigma}(\boldsymbol{x})
    \\
    &= \int_{\Omega} \nabla \cdot\left(  w (\mu_n n^{(0)}
    \nabla\varphi_n^{(2)}   +  \mu_p p^{(0)}\nabla \varphi_p^{(0)}  \right)  d \boldsymbol{x} 
        \\
    &= \int_{\Omega} \nabla   w  \cdot (\mu_n n^{(0)}
    \nabla\varphi_n^{(2)}   +  \mu_p p^{(0)}\nabla \varphi_p^{(0)}) d \boldsymbol{x}
\end{align*}
in which we have considered the first equation in \eqref{eq:varphin2-uD2} and \eqref{eq:order0-holes-bis}. This implies that 
\be
{u}_D^{(2)}   =  -{\mathcal{R}}\displaystyle\displaystyle\left(\int_{\Omega}   \mu_n n^{(0)}
    \nabla\varphi_n^{(2)} \cdot \nabla   w \, d\boldsymbol{x}  +  \int_{\Omega} \mu_p p^{(0)}\nabla \varphi_p^{(0)} \cdot \nabla   w \,  d \boldsymbol{x}\right).
    \label{eq:uD-on-Omega}
\ee 
We define the function 
\be 
\varphi_n^* = \varphi_n^{(2)}  - u_D^{(2)} w
\label{eq:phi-n*}
\ee
which is the solution of the problem 
\be
\begin{cases}
- \nabla\cdot ( \mu_\n n^{(0)} \nabla \varphi_n^{*}  ) =  
   r_0(n^{(0)})(n^{(0)}p^{(0)}-1) - G(\boldsymbol{x}),
\quad \text{ in } \Omega,
\\[1ex]
\frac{
  \partial\varphi_n^{*}}{\partial\boldsymbol{\nu}}=0 \quad \text{ on } \Gamma_{N}, 
\\[1ex]
\varphi_n^{*} =0 \quad \text{ on } \Gamma_{D},
\end{cases}
\label{eq:electrons-exp-star0}
\ee  
since $u_D  \nabla\cdot ( \mu_\n n^{(0)} \nabla w) =0 $ on the domain $\Omega$,  on $\Gamma_{D,1}$ both $\varphi_n^{(2)}$  and $w_D$ vanish, while  on $\Gamma_{D,2}$ the right hand side vanishes  since $\varphi_n^{(2)} =u_D^{(2)} $  and $w_D = 1$. Again, by standard results it can be proved that there exists a unique solution $\varphi_n^{*} \in H^1(\Omega) \cap L^\infty(\Omega) $ of problem \eqref{eq:electrons-exp-star0} such that 
\be
\min\{ 0,  \underline{r} -\overline{G}\} \le 
\varphi_n^{*} \le  \overline{r},
\ee
where $\underline{r}$ and $\overline{r}$ are defined in \eqref{eq:r-overbar-underbar}.

Using \eqref{eq:uD-on-Omega} and \eqref{eq:phi-n*}, we can estimate the function $u_D^{(2)}$ in the following way
\be
\begin{split}
{u}_D^{(2)}  & =  -{\mathcal{R}}\displaystyle\displaystyle\left(\int_{\Omega}   \mu_n n^{(0)}
    \nabla (\varphi_n^{*}  + u_D^{(2)} w) \cdot \nabla   w \, d\boldsymbol{x}  +  \int_{\Omega} \mu_p p^{(0)}\nabla \varphi_p^{(0)} \cdot \nabla   w \,  d \boldsymbol{x}\right)
    \label{eq:uD-on-Omega-estimate}
    \\
    & =  -\displaystyle\left( {\mathcal{R}} \int_{\Omega}   \mu_n n^{(0)}
    |\nabla w |^2 \, d\boldsymbol{x} \right)  u_D^{(2)}    
     \\
    &\quad -\displaystyle {\mathcal{R}}   \int_{\Omega} ( \mu_n n^{(0)}\nabla \varphi_n^{*} + \mu_p p^{(0)}\nabla \varphi_p^{(0)}) \cdot \nabla   w \,  d \boldsymbol{x} 
    \end{split}
\ee 
which implies 
\[
{u}_D^{(2)} = -  \frac{{\mathcal{R}} \displaystyle   \int_{\Omega} ( \mu_n n^{(0)}\nabla \varphi_n^{*} + \mu_p p^{(0)}\nabla \varphi_p^{(0)}) \cdot \nabla   w \,  d \boldsymbol{x} }{1+{\mathcal{R}}  \displaystyle  \int_{\Omega}   \mu_n n^{(0)}
    |\nabla w |^2 \, d\boldsymbol{x}  }
\] 
thus, estimate \eqref{eq:estimate-uDbar} holds. We can conlude that $\varphi_n^{(2)} =\varphi_n^*+u_D^{(2)} w   \in H^1(\Omega) \cap L^\infty(\Omega) $ and 
\be
\min\{ 0,  \underline{r} -\overline{G} -\overline{{u}}_D \} \le \varphi_n^{(2)} \le 
\overline{r} +\overline{{u}}_D,
\ee
with $\overline{{u}}_D$ defined in \eqref{eq:estimate-uDbar}. 

Finally, we can conclude that the linear problem \eqref{eq:psi2alone} admits a unique solutions $\psi^{(2)}\in H^1(\Omega) \cap L^\infty(\Omega)$, this solution satisfies the estimate in \eqref{eq:final-estimate-psi2}.

\end{proof}

\section{Conclusion}
\label{sec:conclusion}
Summarizing the results in the previous section, namely, equations \eqref{eq:order0-uD-phin} and \eqref{eq:order1-uD-phin}, the expansion in \eqref{exp1} reduces to
\beaa
\psi &=& \psi^{(0)} +\delta^2 \psi^{(2)} + O(\delta^3),  
\\
\varphi_n  & = & \varphi_{0} + \delta^2 \varphi_n^{(2)} + O(\delta^3),
\\
\varphi_p & =&  \varphi_p^{(0)}   + O(\delta),
\\
u_D  & = & \delta^2 u_D^{(2)} + O(\delta^3),
\eeaa
and the system \eqref{eq:vR-scaled-model_fin} is asymptotically equivalent to the following decoupled boundary value problems for the leading order functions $\psi^{(0)}$ and $\varphi_p^{(0)}$:
\bea
&&
\begin{cases}
-\lambda^2 \nabla \cdot \left(\varepsilon \nabla \psi^{(0)}\right) = 
- e^{\psi^{(0)}-\varphi_0}+C(\boldsymbol{x}),
\\[1ex] 
\displaystyle \frac{\partial {\psi}^{(0)}}{\partial \boldsymbol{\nu}}= 0,  \quad \text{on }\Gamma_N,
 \\[2ex]
 {\psi}^{(0)} = {\psi}_0 , \quad \text{on }\quad\Gamma_{D_1}\cup \Gamma_{D_2}, 
\label{eq:system-order0-psi0-concl}
\end{cases}
\\[1ex]
&&
\begin{cases}
- \nabla\cdot \left(\mu_\p    e^{\varphi_p^{(0)}-\psi^{(0)}} \nabla \varphi_{\p}^{(0)}\right) =   G(\boldsymbol{x})-r_{0}\left(  e^{\psi^{(0)}-\varphi_0}\right)\left(e^{\varphi_p^{(0)}-\varphi_0}-1\right), 
\\[1ex] 
  \displaystyle\frac{\partial {\varphi}_\p^{(0)}}{\partial \boldsymbol{\nu}}=0,  \quad \text{on }\Gamma_N,
 \\[2 ex]
  {\varphi}_\p^{(0)} ={\varphi}_0, \quad \text{on }\quad\Gamma_{D_1}\cup \Gamma_{D_2}, 
  \label{eq:system-order0-phip0-concl}
\end{cases}
\eea
which establish the leading order carrier densities $n^{(0)}=\exp (\psi^{(0)}-\varphi_0)$, $p^{(0)}=\exp (\varphi_p^{(0)}-\psi^{(0)})$;
to the following coupled problem for the second order correction of the electron quasi-Fermi potential $\varphi_n^{(2)}$ and the current $u_D^{(2)}$:
\bea
&&
\begin{cases}
- \nabla\cdot ( \mu_\n n^{(0)} \nabla \varphi_\n^{(2)}  )= r_0(n^{(0)})(n^{(0)}p^{(0)}-1) -  G(\boldsymbol{x})
\quad \text{ in } \Omega,
\\[1ex]
\frac{
  \displaystyle\partial\varphi_\n^{(2)}}{\partial\boldsymbol{\nu}}=0 \quad \text{ on } \Gamma_{N}, 
  \\[1ex]
\varphi_\n^{(2)} =0 \quad \text{ on } \Gamma_{D,1}, 
\qquad
\varphi_\n^{(2)} =u_D^{(2)} \quad \text{ on } \Gamma_{D,2},
 \\[1ex]
{u}_D^{(2)}   =  -{\mathcal{R}}\displaystyle\displaystyle\int_{{\Gamma}_{D_2}}  \left(\mu_n n^{(0)}\frac{\partial \varphi_n^{(2)}}{\partial \boldsymbol{\nu}} +  \mu_p p^{(0)}\frac{\partial \varphi_p^{(0)}}{\partial \boldsymbol{\nu}} \right) d{\sigma}(\boldsymbol{x}),
\end{cases}
\label{eq:system-order1-concl}
\eea
and the following boundary-value problem for the second order correction of the potential $\psi^{(2)}$:
\bea
&&\begin{cases}
-\lambda^2 \nabla \cdot (\varepsilon \nabla \psi^{(2)}) = 
p^{(0)}-n^{(0)}(\psi^{(2)}-\varphi_n^{(2)}),
\\[1ex]
\displaystyle\frac{\partial {\psi}^{(2)}}{\partial \boldsymbol{\nu}}=
0,  \quad \text{on }\Gamma_N,
\\[1ex]
{\psi}^{(2)}    = 0,\quad \text{on }\quad\Gamma_{D_1}, \qquad 
{\psi}^{(2)} =  {u}_{D}^{(2)},  \quad \text{on }\quad\Gamma_{D_2},
  \end{cases}
  \label{eq:psi2alone-concl}
\eea
which provides the second order correction to the electron carrier density
$n^{(2)}=n^{(0)}(\psi^{(2)}-\varphi_n^{(2)})$ --- recall that $n=n^{(0)}+\delta^2 n^{(2)}$, $p=\delta^2 p^{(0)}$.


A final comment on the effectiveness of the final reduced model. It comprises two decoupled nonlinear  equations \eqref {eq:system-order0-psi0-concl} and  \eqref{eq:system-order0-phip0-concl} and two linear equations. The coupling to the circuit appears in the linear equation \eqref{eq:system-order1-concl}.
For this model, we were able to show existence and uniqueness without having to rely on a smallness assumption for the laser generation term $G$.
We expect the reduced model to be valuable for a theoretical study of the inverse problem, which is the natural setting for the LPS method.

\section*{Acknowledgments}
This work was supported by the Leibniz competition 2020 (NUMSEMIC, J89/2019).

\appendix
\section{Parameter lists}
\label{sec:parameters}
We briefly list the physical parameters for silicon and gallium arsenide. More details can be found in  \cite{Farrell2021} and \cite{Kayser2021}, respectively. The following optical parameters agree in both cases.

 \begin{center}
 \begin{tabular}{l|c|c|c}
 Physical Quantity & Symbol& Value & Units\\
 \hline
 Reference temperature & $T$ & 300 & \si{K}\\
  Thermal voltage & ${V}_{\text{th}}$ & 300 & \si{K}\\
  Laser power & $P_{\si{L}}$ & {0}--{20} & \si{mW}\\
Laser wave length & $\lambda_{\si{L}}$ & {685} & \si{nm}\\
 Laser penetration depth & $d_{A}$ & {4.8} & \si{\micro\meter}\\
 Laser spot radius & $\sigma_{{L}}$ & $\geq$0.02585 & \si{V}
 \end{tabular}
 \end{center}
\subsection{Parameters for silicon}

 \begin{center}
 \begin{tabular}{l|c|c|c}
 Physical Quantity & Symbol& Value & Units\\
 \hline
 Band gap  & $E_{g}$ & \SI{1.12}{}  & \si{eV}\\
 Density of states in the conduction band & $N_{c}$ & \SI{1.04E19}{} & \si{1/cm^3}\\
 Density of states in the valence band & $N_{v}$ & \SI{2.8E19}{} & \si{1/cm^3}\\
 Relative permittivity& $\varepsilon_{\text{Si}}$ & \SI{11.8}{} & -\\
Reference mobility value  & $\bar{\mu}=\mu_{\si{n,0}}^{\si{ref}}$ & \SI{1323}{} & \si{{cm^2/Vs}}\\
Reference doping profile value  & $\bar{C}$ & \SI{1.2E16}{} & \si{1/cm^3}
 \end{tabular}
 \end{center}

\subsection{Parameters for gallium arsenide}

 \begin{center}
 \begin{tabular}{l|c|c|c}
 Physical Quantity & Symbol& Value & Units\\
 \hline
 Band gap  & $E_{g}$ & \SI{1.424}{}  & \si{eV}\\
 Density of states in the conduction band & $N_{c}$ & \SI{4.7E17}{} & \si{1/cm^3}\\
 Density of states in the valence band & $N_{v}$ & \SI{9E18}{} & \si{1/cm^3}\\
Relative permittivity& $\varepsilon_{\text{SiGa}}$ & \SI{12.9}{} & -\\
 Reference mobility value  & $\bar{\mu}=\mu_{\si{n,0}}^{\si{ref}}$ & \SI{9400}{} & \si{{cm^2/Vs}}\\
Reference doping profile value  & $\bar{C}$ & \SI{1.2E18}{} & \si{1/cm^3}
 \end{tabular}
 \end{center}

\bibliography{references.bib}

\end{document}